\documentclass[twocolumn]{IEEEtran}
\usepackage{cite}
\usepackage{graphicx}
\usepackage[cmex10]{amsmath}
\usepackage{amsfonts}
\usepackage{amssymb}
\usepackage{amsbsy}
\usepackage{amsthm}
\usepackage[mathscr]{eucal}
\usepackage[normal]{subfigure}
\usepackage{algorithm}
\usepackage{algorithmic}
\usepackage{comment}
\usepackage{paralist}
\usepackage{boxedminipage}
\usepackage{array}
\usepackage{multirow}
\usepackage{float}
\usepackage{color}

\newtheorem{theorem}{Theorem}

\begin{document}
\title{Experiment, Modeling, and Analysis\\ of Wireless-Powered Sensor Network\\ for Energy Neutral Power Management}

\author{Dedi Setiawan, Arif Abdul Aziz, Dong In Kim, \IEEEmembership{Senior Member, IEEE} and Kae Won Choi, \IEEEmembership{Senior Member, IEEE}
\thanks{D.~Setiawan is with the Convergence Institute of Biomedical Engineering \& Biomaterials, Seoul National University of Science and Technology, Korea (email: morethanubabe@gmail.com).}
\thanks{A.~A.~Aziz is with the Dept.~of Computer Science and Engineering, Seoul National University of Science and Technology, Korea (email: arif.abdul.aziz92@gmail.com).}
\thanks{D.~I.~Kim and K.~W.~Choi are with the School of Information and Communication Engineering, Sungkyunkwan University (SKKU), Suwon, Korea (email: dikim@skku.ac.kr and kaewon.choi@gmail.com).}
\thanks{This work was supported by the National Research Foundation of Korea (NRF) grant funded by the Korean government (MSIP) (2014R1A5A1011478).}
}

\maketitle

\begin{abstract}
In this paper, we provide a comprehensive system model of a wireless-powered sensor network (WPSN) based on experimental results on a real-life testbed.
In the WPSN, a sensor node is wirelessly powered by the RF energy transfer from a dedicated RF power source.
We define the behavior of each component comprising the WPSN and analyze the interaction between these components to set up a realistic WPSN model from the systematic point of view.
Towards this, we implement a real-life and full-fledged testbed for the WPSN and conduct extensive experiments to obtain model parameters and to validate the proposed model.
Based on this WPSN model, we propose an energy management scheme for the WPSN, which maximizes RF energy transfer efficiency while guaranteeing energy neutral operation.
We implement the proposed energy management scheme in a real testbed and show its operation and performance.
\end{abstract}


\section{Introduction}\label{section:introduction}

The wireless power transfer technology enables a power source to wirelessly transfer electrical energy to another device by means of electromagnetic fields.
While near-field wireless power transfer technologies such as inductive and magnetic resonant coupling can deliver energy to only close-by devices, the radio frequency (RF) energy transfer technology is able to realize far-field power transfer to remotely located devices \cite{Huang:2015}.
However, the end-to-end efficiency of the RF energy transfer is typically very low since a receive antenna captures only a very small fraction of electromagnetic energy that is spread out in space.
Therefore, one of the most promising application areas of RF energy transfer is sensor networks since sensor nodes are generally sustainable by a relatively small amount of energy provision \cite{Xie:2013}.
In this paper, we focus on studying a wireless-powered sensor network (WPSN) that is powered by the RF energy transfer from a dedicated RF power source.

In recent years, a number of research works have been published by communications society in the area of wireless-powered communication networks (WPCNs) \cite{Bi:2015,Lu:2016}.
The WPCN is defined as wireless networks consisting of a hybrid access point (H-AP), which acts as both a communication gateway and a power source, and communication nodes, which are powered by the RF energy transfer from the H-AP.
Most of the works regarding the WPCN focuses on theoretical investigation into radio resource allocation (e.g., \cite{Lu:2015Dec}), beamforming (e.g., \cite{Choi:2015}), cooperative communications (e.g., \cite{Chen:2015}), and full-duplex communications (e.g., \cite{Kang:2015}).
The weakness of these theoretical works is that they do not incorporate realistic power transfer, energy harvesting, and energy consumption models.
For example, most of these works ideally assume that a dominant cause of energy drain is the power used for data transmission, while the circuit power consumption is more significant in most practical applications.
In addition, these works assume that RF energy harvesting efficiency is constant even though the actual RF energy harvesting efficiency has a complex non-linear curve because of non-ideal diode behavior.
Therefore, it is of paramount importance to set up a practical WPCN model based on real experiments on a testbed.

Other than theoretical studies on the WPCN, practical circuit design issues of the RF energy transfer have been investigated in \cite{Dolgov:2010, Popovic:2013, Visser:2012, Farinholt:2009}.
In these works, the authors design and implement RF energy harvesting circuit components such as a rectenna and a maximum power point tracking (MPPT) module.
Among these works, \cite{Dolgov:2010} and \cite{Popovic:2013} target to harvest energy from ambient power sources (e.g., TV and cell towers) whereas a dedicated RF power source is considered in \cite{Visser:2012} and \cite{Farinholt:2009}.
Since these works focus their efforts on designing individual circuit components, they do not give a systematic and integrated view of sensor networks powered by the RF energy transfer.

In this paper, we provide a comprehensive system model of WPSNs based on real-life experiments.
The WPSN under consideration consists of a power beacon (i.e., an RF power source) and a sensor node.
The power beacon and the sensor node again consist of many circuit components, for example, an amplifier, an RF transceiver, a wireless energy harvester, and an energy storage, etc.
Rather than delving into each component, we model the operation of the whole WPSN system by defining the behavior of each component and the interaction between components.
Moreover, we set up a real-life WPSN testbed and conduct extensive experiments to obtain model parameters and validate the proposed model.
To our best knowledge, there has been no research work that provides an experiment-based model of the WPSN.
Therefore, our WPSN model can be used as a baseline model for computer simulation and theoretical analysis of the WPSN.
In addition, this model is expected to facilitate the understanding of the WPSN and to highlight potentials and limitations of the WPSN.

Another contribution of this paper is to propose an energy management scheme for energy neutral operation based on the WPSN model.
The energy neutral operation makes sure that power consumption is less than harvested power for keeping a sensor node alive.
A large number of research works have been conducted on energy harvesting sensor networks harnessing energy from ambient energy sources \cite{Sudevalayam:2011}.
Since ambient energy harvesting is typically uncontrollable, an energy management scheme controls a duty cycle of a sensor node to guarantee energy neutral operation.
For example, adaptive energy management schemes are proposed for sensor networks exploiting solar power (e.g., \cite{Moser:2010, Renner:2014}) and ambient RF power (e.g., \cite{Shigeta:2013, Vyas:2013}).
In contrast to these works assuming uncontrollable ambient power sources, our WPSN model adopts a controllable dedicated RF power source.
This controllability of the power source adds a new dimension to the traditional energy neutral operation problem.
The proposed energy management scheme aims at minimizing power consumption of an amplifier in the RF power source while guaranteeing the energy neutral operation.
The proposed energy management scheme achieves its goal by adaptively controlling the RF transmit power so that the RF energy transfer efficiency is maximized.
We implement the proposed energy management scheme in a real testbed and show its operation and performance.

In our companion work \cite{Choi:2016}, we have studied the multi-antenna WPSN, in which the power beacon concentrates the RF energy on the sensor node by using an antenna array to enhance the RF energy transfer efficiency.
The main focus of this companion work is to propose an energy beamforming algorithm that dynamically steers the microwave beam and to conduct experiments on the multi-antenna WPSN testbed for testing the algorithm performance.
Differently from the companion work, this paper aims to set up an integrated system model of the single antenna WPSN and to propose an efficient energy management scheme.

The rest of the paper is organized as follows.
We present the WPSN system model in Section \ref{section:system model}.
In Section \ref{section:measurement}, we explain our testbed implementation and present measurement results.
The proposed energy management scheme is described in Section \ref{section:energy management}.
Section \ref{section:result} presents detailed experimental results, and the paper is concluded in Section \ref{section:conclusion}.

\section{System Model}\label{section:system model}

\subsection{Wireless-Powered Sensor Network Model}

\begin{figure}
	\centering
    \includegraphics[width=7cm, bb=1.6in 2.3in 8.4in 6in]{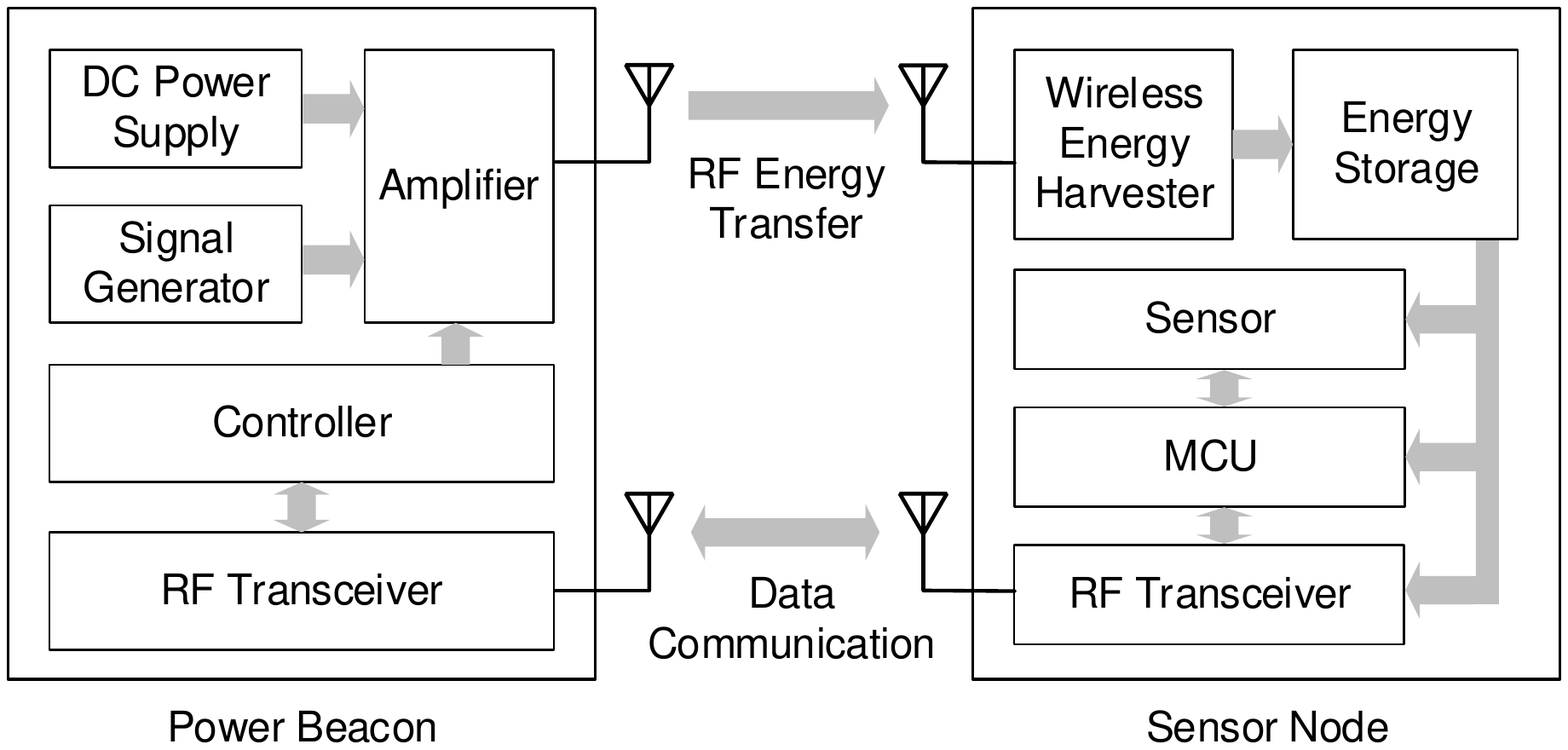}
    \caption{Block diagram of WPSN model.}
    \label{fig:model}
\end{figure}

In this section, we introduce an overall system model of the WPSN under consideration.
In Fig.~\ref{fig:model}, we show the block diagram of the WPSN model.
The WPSN consists of one power beacon and one sensor node.
The basic role of the power beacon is to wirelessly supply energy to the sensor node by means of the RF energy transfer technique.
The power beacon is connected to a power grid.
On the other hand, the sensor node relies solely on the energy wirelessly supplied by the power beacon, without any other power source.
The power beacon gathers sensor measurement results from the sensor node and controls the operation of the sensor node through a data communication link.

In the power beacon, the signal generator generates a weak continuous wave (CW) RF source signal and feeds the RF source signal to the amplifier.
The amplifier amplifies the weak RF source signal from the signal generator and sends a high-power RF signal to the air through the transmit antenna.
For amplification, the amplifier consumes DC power from the DC power supply.
Therefore, the amplifier performs DC-to-RF power conversion.
The controller of the power beacon conducts a power management function that keeps the stored energy of the sensor node to a sufficient level while minimizing power consumption of the amplifier.
For the power management, the controller is able to control the gain of the amplifier.
In addition, the controller can send a control command to the sensor node via the RF transceiver for the power management.
The control decision made by the controller is based on the sensor measurement report from the sensor node, which is delivered by the RF transceiver.

In the sensor node, the wireless energy harvester receives the RF signal from the power beacon through the receive antenna.
The wireless energy harvester performs RF-to-DC conversion by using a rectifier.
The DC power from the wireless energy harvester is stored in the energy storage.
We adopt a supercapacitor as the energy storage rather than a rechargeable battery.
A supercapacitor is considered as a suitable energy storage device for energy harvesting sensor nodes since it can endure rapid charging and discharging cycles.
The micro controller unit (MCU), which consists of a central processing unit (CPU) and peripherals, controls the whole sensor node.
The MCU obtains sensing results from the sensors and sends the sensor measurement report to the power beacon via the RF transceiver.
The MCU, the RF transceiver, and the sensor consume energy stored in the energy storage.

The WPSN system model has two RF channels: the RF energy transfer channel and the data communication channel.
These two channels use different frequency bands, and therefore there is no interference between these channels.

\subsection{Power Beacon and RF Energy Transfer Channel Model}

In this subsection, we provide a more detailed model of the power beacon and the RF energy transfer channel.
The RF source signal generated by the signal generator is a CW signal with the frequency $f_o$ and the power $p_\text{src}$.
The frequency $f_o$ and the source power $p_\text{src}$ are fixed.
The amplifier amplifies this weak RF source signal to generate a transmit RF signal.
Henceforth, the power of the transmit RF signal will be called transmit power.
We assume that the amplifier can be dynamically turned on and off by the controller.
The transmit power, denoted by $p_\text{tx}$, is given by
\begin{align}\label{eq:amptx}
p_\text{tx} = \chi_\text{amp}\cdot g_\text{amp}(p_\text{src})\cdot p_\text{src},
\end{align}
where $g_\text{amp}(p)$ is the amplifier gain as a function of amplifier input power $p$, and $\chi_\text{amp}$ is the indicator that is one if the amplifier is turned on; and zero, otherwise.
The transmit power cannot exceed the maximum output power of the amplifier.

The power amplifier consumes power from the DC power supply.
Let $p_\text{cons}$ denote the amplifier power consumption.
When the amplifier is turned on, the power added efficiency (PAE) of the amplifier is given by
\begin{align}\label{eq:pae}
\text{PAE} = \frac{p_\text{tx}-p_\text{src}}{p_\text{cons}} \simeq \frac{p_\text{tx}}{p_\text{cons}},
\end{align}
since $p_\text{src}$ is negligibly small compared to $p_\text{tx}$.
Since the PAE is a function of the output power of the amplifier (i.e., the transmit power), we define $\theta(p)$ as the PAE function mapping the output power $p$ to the PAE.
The amplifier power consumption is zero when the amplifier is turned off (i.e., when $\chi_\text{amp}=0$).
From \eqref{eq:pae}, the amplifier power consumption is given as a function of $p_\text{tx}$ by
\begin{align}\label{eq:ampcons}
p_\text{cons} = \chi_\text{amp}\cdot \frac{p_\text{tx}}{\theta(p_\text{tx})}.
\end{align}

The transmit RF signal from the amplifier is sent from the transmit antenna of the power beacon.
This transmit RF signal goes through the RF energy transfer channel and arrives at the receive antenna of the sensor node.
The receive power, denoted by $p_\text{rx}$, is defined as the power of the received RF signal at the sensor node.
The receive power is given by
\begin{align}
p_\text{rx} = h\cdot p_\text{tx},
\end{align}
where $h$ is the power attenuation of the RF energy transfer channel.
The power attenuation depends on the distance between the power beacon and the sensor node, which is denoted by $d$.
The power attenuation function $\psi$ defines the relationship between the power attenuation and the distance according to the following path loss formula:
\begin{align}
h=\psi(d) = G/d^\nu,
\end{align}
where $d$ is the distance in meter, $\nu$ is a path loss exponent, and $G$ is the power attenuation at one meter distance.

\subsection{Sensor Node Power Model}\label{section:powermodel}

In this subsection, we explain the sensor node model in more detail.
Fig.~\ref{fig:circuit} shows an equivalent power circuit model of the sensor node.
This model consists of three hardware modules: a wireless energy harvesting module, an energy storage module, and a sensor module.
The RF signal from the power beacon is received by the wireless energy harvesting module and is converted to a DC current to supply energy to the sensor node.
A supercapacitor is used in the energy storage module.
The harvested energy is stored in the supercapacitor when more energy is harvested than the energy consumed by the sensor module.
The sensor module obtains sensing results, performs computation, and sends the sensing results to the power beacon.
For this operation, the sensor module consumes energy harvested by the wireless energy harvesting module.

\begin{figure}
	\centering
    \includegraphics[width=6.5cm, bb=1in 4.8in 7.4in 9.3in]{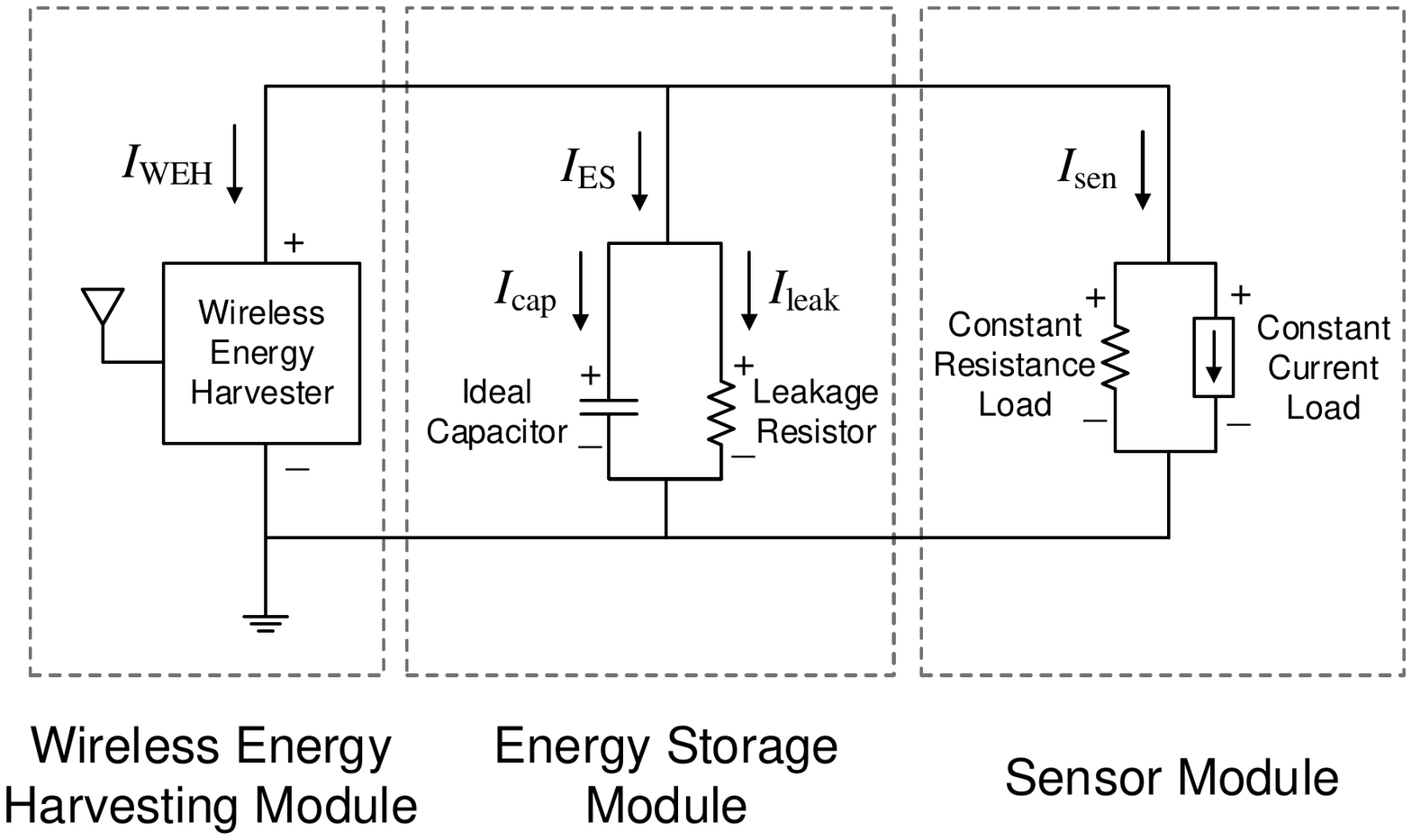}
    \caption{Sensor node equivalent power circuit model.}
    \label{fig:circuit}
\end{figure}

In the sensor node model in Fig.~\ref{fig:circuit}, all circuit components are simply connected in parallel and have the same voltage, denoted by $V$, dropped across them.
Henceforth, we will call $V$ the sensor node voltage.
The sensor node voltage is dependent upon the voltage across the supercapacitor, which is in turn decided by the stored energy in the supercapacitor.
In our sensor node model, we do not incorporate a DC-DC converter component such as an MPPT module and the voltage regulator.
Although an MPPT module can apply the proper resistance to the wireless energy harvester to obtain the maximum power, an MPPT module consumes some energy by itself.
Since only a very small amount of energy can be harvested by the wireless energy harvester, it is not appropriate to use an energy-consuming MPPT module.
For the same reason, we do not use the voltage regulator for the sensor module.
Since we do not adopt any DC-DC converter component, the input voltage to the wireless energy harvester and the sensor module is equal to the voltage across the supercapacitor.

The wireless energy harvester receives the RF signal from the power beacon.
The wireless energy harvester generally consists of a matching network and a rectifier.
The performance of the wireless energy harvester depends on how these components are designed.
Due to the nonlinearity of these components, it is very difficult to derive a meaningful analytic expression for describing the behavior of the wireless energy harvester.
Rather than modeling each component of the wireless energy harvester, a current-voltage curve (i.e., an I-V curve) according to a given received power (i.e., $p_\text{rx}$) can be used for fully characterizing the wireless energy harvester.
The I-V function of the wireless energy harvester, denoted by $\rho$, is defined as
\begin{align}\label{eq:ehiv}
I_\text{WEH} = -\rho(V,p_\text{rx}),
\end{align}
where $p_\text{rx}$ is the received power and $I_\text{WEH}$ is the current through the wireless energy harvester.

As shown in Fig.~\ref{fig:model}, a supercapacitor of the energy storage module is typically modeled by the equivalent circuit consisting of an ideal capacitor and a leakage resistor.
The ideal capacitor has a capacitance of $C$.
The time derivative of the voltage across the ideal capacitor is proportional to the current as
\begin{align}
\frac{\mathrm{d}V}{\mathrm{d}t} = \frac{I_\text{cap}}{C},
\end{align}
where $I_\text{cap}$ is the current through the ideal capacitor.
The stored energy in the ideal capacitor, denoted by $E$, is given by
\begin{align}\label{eq:sten}
E = \frac{1}{2}C V^2.
\end{align}
A leakage resistor is introduced to describe a non-ideal behavior of a real supercapacitor.
A supercapacitor discharges by itself even when it is disconnected from other parts of the circuit.
If the resistance of the leakage resistor is $R_\text{leak}$, the leakage current is $I_\text{leak} = V/R_\text{leak}$.
The supercapacitor in consideration is fully characterized by two parameters: the capacitance $C$ and the leakage resistance $R_\text{leak}$.

The sensor module acts as a load that draws energy from the wireless energy harvester and the supercapacitor.
The main energy sinks in the sensor module are integrated circuits (ICs) such as the MCU and the RF transceiver.
In the sensor module, the RF transceiver is typically compliant to a low power communication standard such as IEEE 802.15.4.
We do not consider energy consumption of sensors such as a temperature sensor since it is very small compared to the energy consumption of other ICs.
The ICs on the sensor module act as different types of loads depending on how they consume energy.
Two types of loads are incorporated in this model: a constant resistance load and a constant current load.
Typically, the sensor module alternates between the several different sensor modes: idle, active, receive, and transmit.
The mode of the sensor module is indexed by $m=\text{idle}$, $\text{act}$, $\text{rx}$, and $\text{tx}$ for the idle, active, receive, and transmit mode, respectively.
Almost all the ICs are turned off except for minimal functionalities (e.g., timer) in the idle mode, the MCU is activated in the active mode, the RF transceiver is ready to receive data in the receive mode, and the RF transceiver transmits data in the transmit mode.
Each mode has different load characteristics.
When the sensor module is in mode $m$, the resistance of the constant resistance load and the current of the constant current load are denoted by $\gamma(m)$ and $\zeta(m)$, respectively.

The sensor module stops working if the input voltage drops under the minimum voltage required for operation.
We define this minimum voltage as the minimum sensor node voltage, denoted by $V_\text{min}$.
The wireless energy harvester cuts off the current if the sensor node voltage exceeds the maximum allowed voltage of the supercapacitor or the sensor module to avoid damage to these components.
We define this maximum voltage as the maximum sensor node voltage, denoted by $V_\text{max}$.

\subsection{Stored Energy Evolution Model}\label{section:analysis}

This subsection introduces the derivation of the ordinary differential equation (ODE) that governs the time evolution of the stored energy $E$ in the sensor node.
The time evolution model of the stored energy in a supercapacitor can be obtained by analyzing the equivalent circuit
\cite{Mishra:2015,Renner:2014}.
It is noted that the stored energy evolution model in this subsection and the sensor node power model in Section \ref{section:powermodel} are partly introduced in \cite{Choi:2016} as well.
However, these models are not experimentally validated in \cite{Choi:2016}.

The time derivative of the stored energy in \eqref{eq:sten} is equal to the power charged to the ideal capacitor as given in the following equation:
\begin{align}\label{eq:dedt}
\frac{\mathrm{d}E}{\mathrm{d}t} = VI_\text{cap}.
\end{align}
By the Kirchhoff's law, the current through the ideal capacitor $I_\text{cap}$ in \eqref{eq:dedt} is given by
\begin{align}
I_\text{cap} = I_\text{ES} - I_\text{leak}.\label{eq:i1}
\end{align}

Let us obtain the current through the energy storage module $I_\text{ES}$ and the current through the leakage resistor $I_\text{leak}$.
By the Kirchhoff's law, the current through the energy storage module is given by
\begin{align}
I_\text{ES} = - I_\text{WEH} - I_\text{sen},\label{eq:i2}
\end{align}
The current through the wireless energy harvester is $I_\text{WEH} = -\rho(V,p_\text{rx})$ from \eqref{eq:ehiv} and the current through the sensor module is given by
\begin{align}\label{eq:isen}
I_\text{sen} = \frac{V}{\gamma(m)} + \zeta(m).
\end{align}
From \eqref{eq:ehiv}, \eqref{eq:i2}, and \eqref{eq:isen}, the current through the energy storage module is
\begin{align}\label{eq:ies}
I_\text{ES} = - I_\text{WEH} - I_\text{sen}
=\rho(V,p_\text{rx}) - \frac{V}{\gamma(m)} - \zeta(m).
\end{align}
In addition, we obtain the current through the leakage resistor $I_\text{leak}$ as
\begin{align}\label{eq:ileak}
I_\text{leak} = \frac{V}{R_\text{leak}}.
\end{align}

From \eqref{eq:dedt}, \eqref{eq:i1}, \eqref{eq:ies}, and \eqref{eq:ileak}, the time derivative of the stored energy is given by
\begin{align}\label{eq:evol1}
\begin{split}
\frac{\mathrm{d}E}{\mathrm{d}t} & = VI_\text{cap} = V(I_\text{ES}-I_\text{leak})\\
&=\rho(V,p_\text{rx})V - \frac{V^2}{\gamma(m)} - \zeta(m) V -  \frac{V^2}{R_\text{leak}}.
\end{split}
\end{align}
In \eqref{eq:evol1}, $\rho(V,p_\text{rx})V$ is the harvested power by the wireless energy harvester.
We define the wireless energy harvesting efficiency as the ratio of the harvested power to the receive power.
Then, the wireless energy harvesting efficiency as functions of $V$ and $p_\text{rx}$ is given by
\begin{align}\label{eq:efficiency}
\eta_V(V,p_\text{rx}) = \frac{\rho(V,p_\text{rx})V}{p_\text{rx}}.
\end{align}

Since $V = \sqrt{2E/C}$ from \eqref{eq:sten}, we can rewrite \eqref{eq:evol1} with respect to the stored energy $E$ as follows:
\begin{align}\label{eq:evol2}
\begin{split}
\frac{\mathrm{d}E}{\mathrm{d}t} = \phi(E,p_\text{rx}) - \xi_m(E) - \xi_\text{leak}(E).
\end{split}
\end{align}
In \eqref{eq:evol2}, $\phi(E,p_\text{rx})$ is the harvested power such that
\begin{align}\label{eq:phieph}
\phi(E,p_\text{rx}) = \eta_E(E,p_\text{rx})\cdot p_\text{rx},
\end{align}
where $\eta_E(E,p_\text{rx})$ is the wireless energy harvesting efficiency function with respect to the stored energy such that
\begin{align}
\eta_E(E,p_\text{rx}) = \eta_V\big(\sqrt{2E/C},p_\text{rx}\big).
\end{align}
In \eqref{eq:evol2}, $\xi_m(E)$ is the sensor module power consumption in mode $m$ such that
\begin{align}
\xi_m(E) = \frac{2}{C\gamma(m)}E + \sqrt{\frac{2 \zeta(m)^2}{C}}\sqrt{E},
\end{align}
and $\xi_\text{leak}(E)$ is the supercapacitor leakage power such that
\begin{align}
\xi_\text{leak}(E) = \frac{2}{C\cdot R_\text{leak}}E.
\end{align}

We define the minimum and maximum stored energies corresponding to the minimum and the maximum sensor node voltages as $E_\text{min} =  C(V_\text{min})^2/2$ and $E_\text{max} =  C(V_\text{max})^2/2$, respectively.
Then, the stored energy $E$ should be kept above $E_\text{min}$ for continuous sensor node operation.
In addition, the stored energy cannot exceed $E_\text{max}$ since no power is charged by the wireless energy harvester in the case that $E \ge E_\text{max}$.

\section{Testbed Setup and Measurement}\label{section:measurement}

\subsection{Testbed Setup}

We set up a real-life WPSN testbed that implements the system model described in Section \ref{section:system model}.
This testbed is used for measuring the parameters and the functions of the system model (i.e., $\theta$, $\psi$, $\rho$, $C$, $R_\text{leak}$, $\gamma(m)$, and $\zeta(m)$), verifying the system model, and evaluating the proposed energy management scheme.
This testbed has the same architecture as the system model in Fig.~\ref{fig:model}, that is, the testbed consists of a power beacon, which is able to wirelessly send energy via microwave, and a sensor node, which makes use of the energy received from the power beacon.
Only off-the-shelf commercially available hardware components are used to build the testbed, and the testbed is controlled by software so that various energy and data transmission experiments can easily be conducted and reproduced.

The power beacon consists of a controller, a signal generator, an amplifier, and a DC power supply. 
In the testbed, the controller is comprised of a desktop computer, a PXI chassis, and an FPGA module.
In addition, a signal generator is Tektronix TSG-4104A, and an amplifier is RFMD RF2173. 
The RF source signal generated by the signal generator is fed into the amplifier.
The amplifier is controlled by the analog output voltage of the FPGA module, which is connected to the desktop computer via the PXI chassis.
In the desktop computer, a Labview software is used to control the power beacon and process the measurement data.
The amplifier draws power from the DC power supply (i.e., GW Instek GPC-3060D) to amplify the RF source signal.
The RF energy transfer uses the frequency of 920 MHz.
Fig.~\ref{fig:pbpic} shows a picture of the power beacon in our testbed.

\begin{figure}
    \centering
    \subfigure[Power beacon]{
        \label{fig:pbpic}\includegraphics[width=2.92cm, bb=0.0in 0in 7.5in 10.5in] {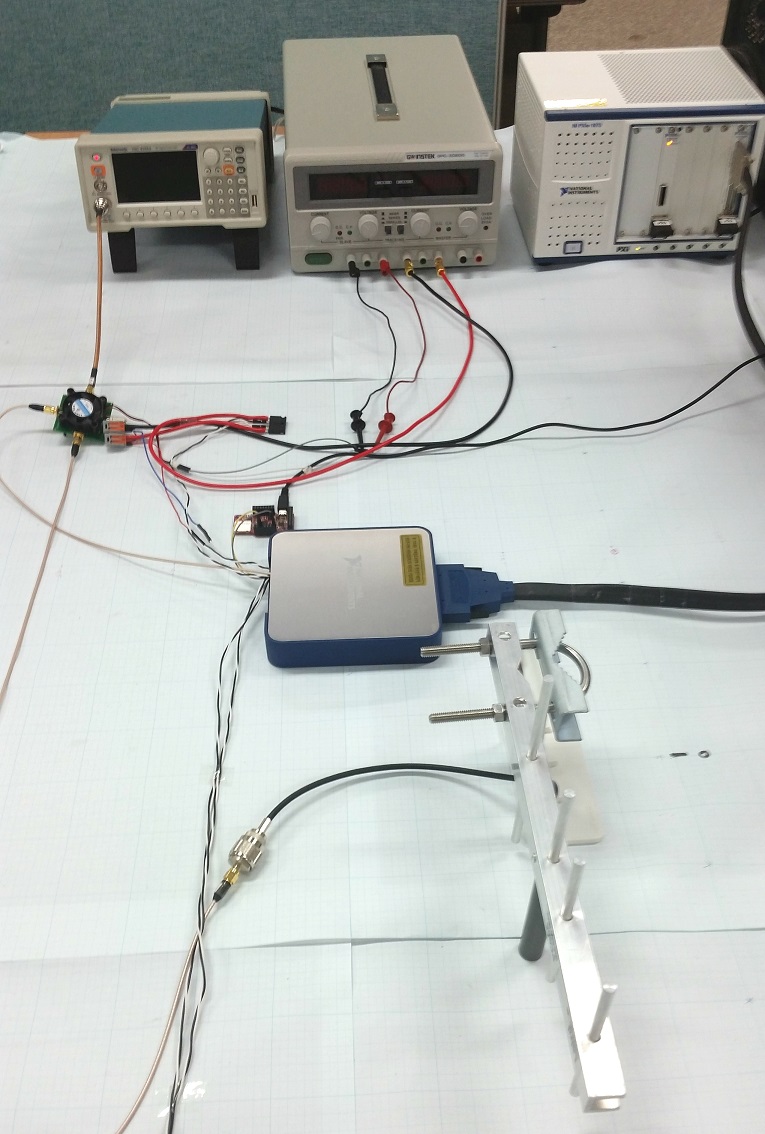}
        }~~~~
    \subfigure[Sensor node]{
        \label{fig:snpic}\includegraphics[width=3.15cm, bb=0.0in 0in 6.3in 10in] {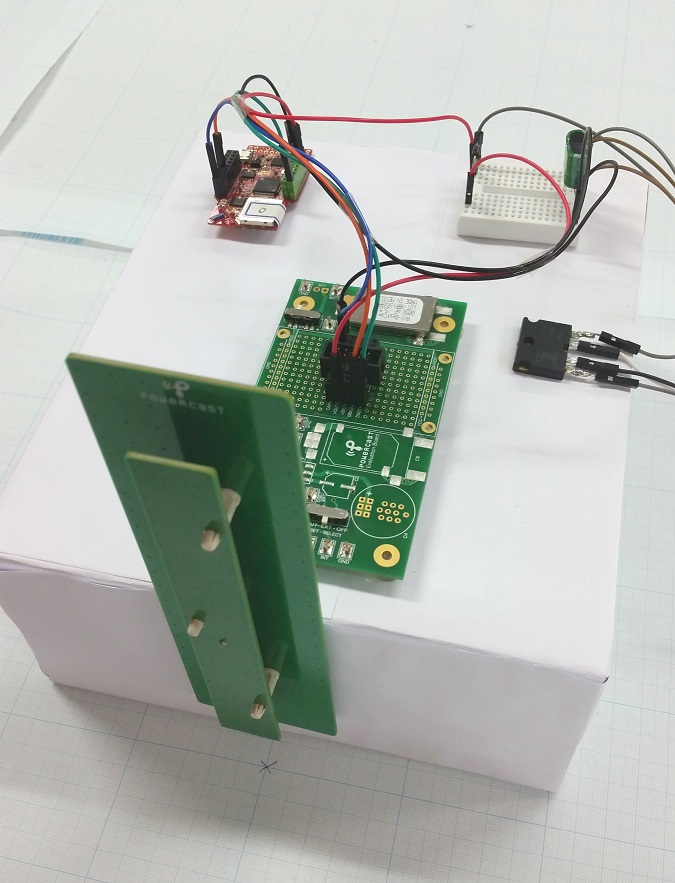}
        }
    \caption{Pictures of power beacon and sensor node in WPSN testbed.}
    \label{fig:testbedpic}
\end{figure}

We use two types of RF energy transfer channels for the testbed.
The first one is the antenna-based channel environment, in which the power beacon sends the RF signal by a Yagi antenna (i.e., Laird PC904N Yagi antenna) with 8 dBi gain, and the sensor node receives the RF signal by a PCB patch antenna with 6.1 dBi gain.
Although we can test actual wireless power transfer over the air in this antenna-based channel environment, it is difficult to obtain reproducible results since accurate control of the power attenuation is not possible.
Therefore, we use a step attenuator-based channel environment, in which the RF signal from the power beacon goes through a step attenuator to get to the sensor node.
In this channel environment, we can accurately control the power attenuation to simulate the real wireless channel with various distances.
Unless noted otherwise, we will mainly use the step attenuator-based channel for reproducibility.

The sensor node in the testbed consists of a wireless energy harvesting module, an energy storage module, and a sensor module with the same arrangement as in Fig.~\ref{fig:circuit}.
In the sensor node, a sensor module (i.e., Zolertia Z1 mote \cite{zolertiaz1}) draws current from a supercapacitor (i.e., Samxon DDL series) that stores energy charged by a wireless energy harvesting module (i.e., Powercast P1110 \cite{p1110}).
The wireless energy harvesting module has an internal RF power sensor and is capable of measuring the receive power.
The sensor module can acquire the received power from the wireless energy harvesting module by using analog-to-digital conversion (ADC).
The sensor module can also measure the sensor node voltage, which indicates the amount of the energy stored in the supercapacitor.
The sensor module adopts TI MSP430 as an MCU and TI CC2420 as an RF transceiver.
In Fig.~\ref{fig:snpic}, we show a picture of the sensor node in our testbed.

The power beacon and the sensor node exchange information by means of the IEEE 802.15.4 RF transceiver on the 2.4 GHz ISM band.
In the sensor node, the IEEE 802.15.4-compliant chip (i.e., TI CC2420) in the sensor module is used for communication.
On the other hand, in the power beacon, a Zolertia Z1 mote attached to the desktop computer via serial connection is used for communication.
By using this communication link, the sensor node can report the receive power and the sensor node voltage to the power beacon.
The power beacon can send a command to the sensor node via this communication link.

\subsection{Power Beacon and RF Energy Transfer Channel Measurement}

\begin{figure}
    \centering
    \subfigure[Amplifier power consumption and PAE]{
        \label{fig:ampcon}\includegraphics[width=4.1cm, bb=1.0in 0.3in 10.4in 7.4in] {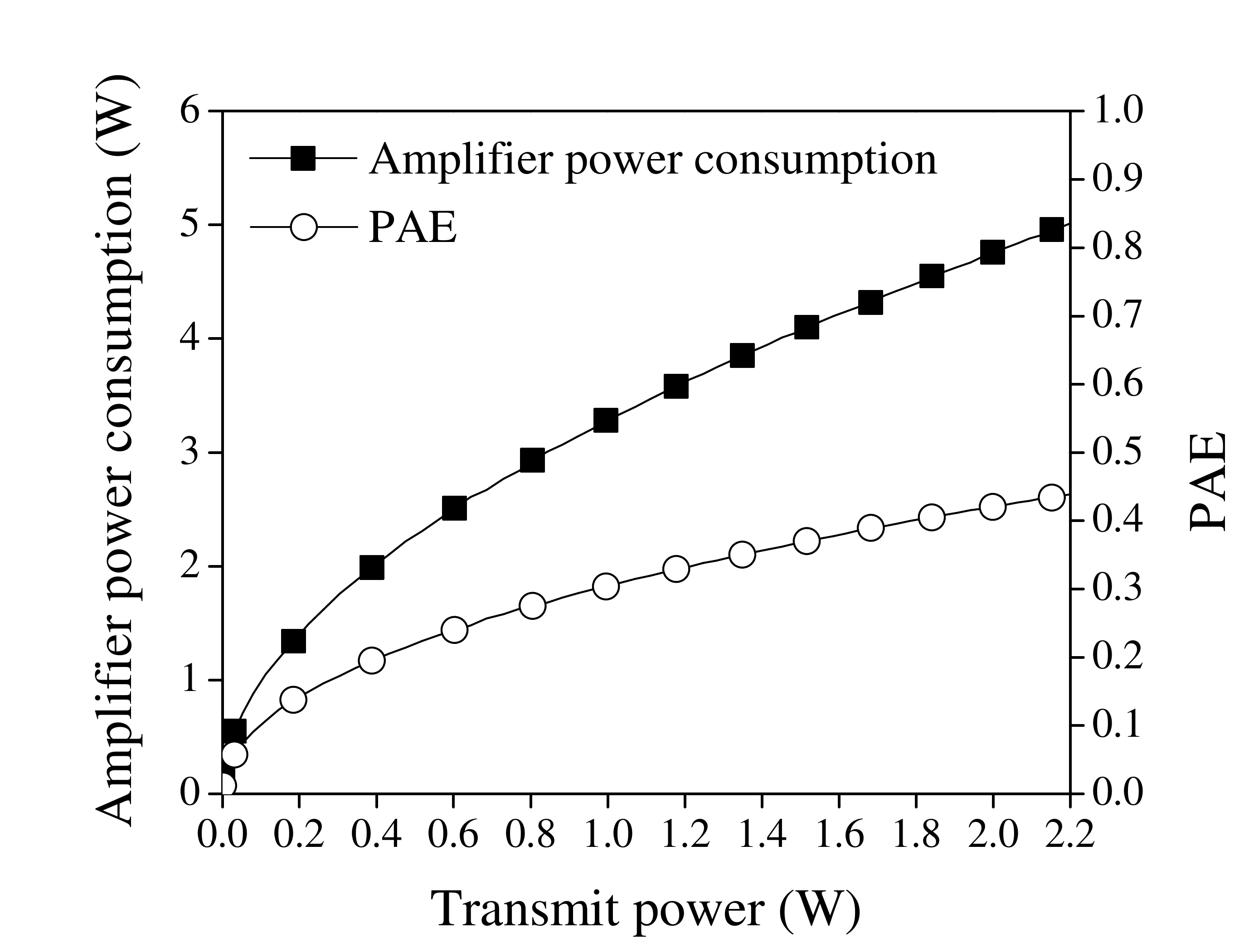}
        }~
    \subfigure[Power attenuation]{
        \label{fig:atten}\includegraphics[width=4.1cm, bb=0.7in 0.3in 9.4in 7.4in] {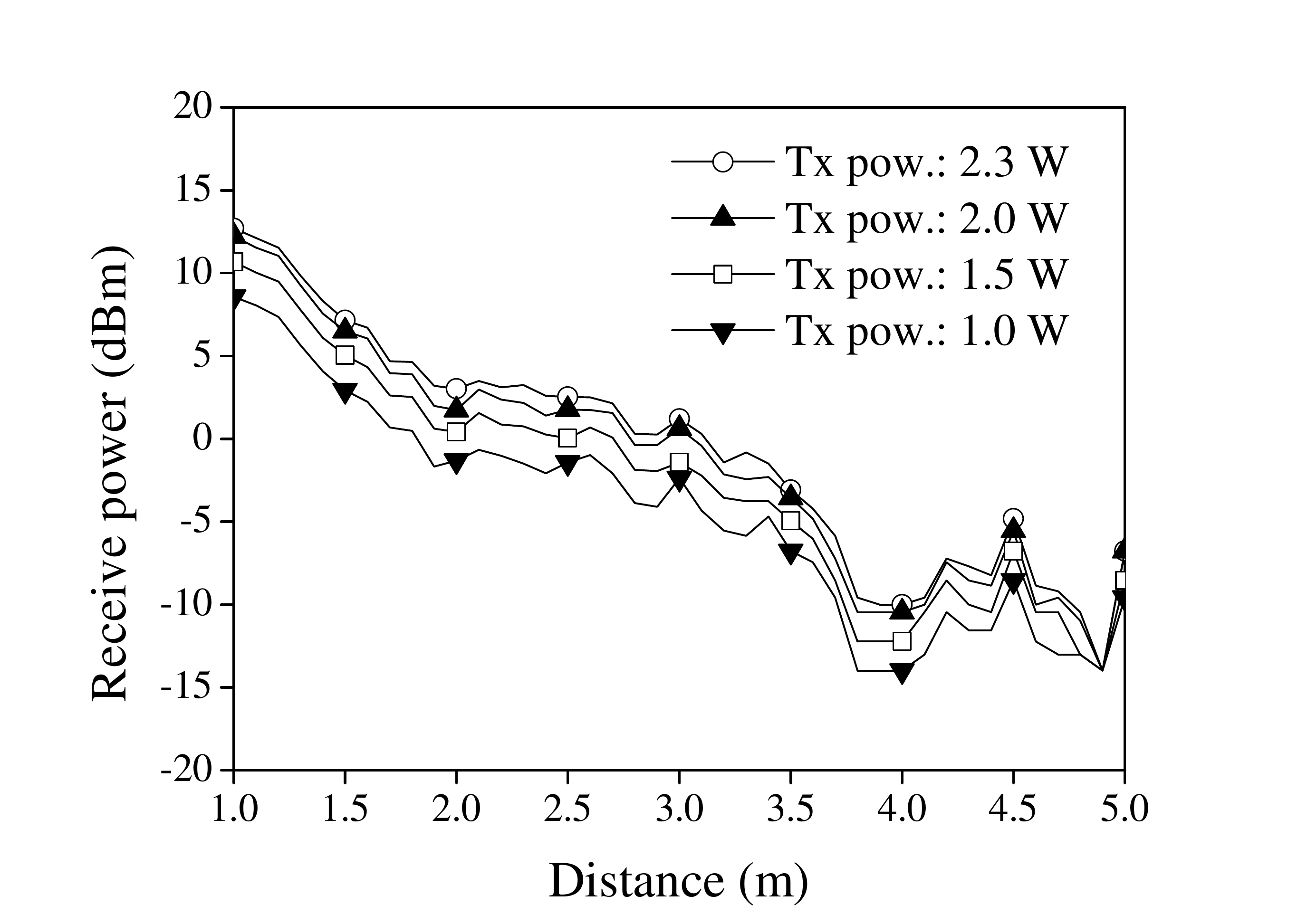}
        }
    \caption{Amplifier power consumption, PAE, and power attenuation.}
    \label{fig:ampconatten}
\end{figure}

We have conducted an experiment on the power amplifier to obtain the amplifier power consumption and the PAE according to the transmit power.
We set the power of the RF source signal to 6 dBm and the voltage of the DC power supply to 3.5 V.
Fig.~\ref{fig:ampcon} shows that the PAE function (i.e., $\theta$) is an increasing function of the transmit power.

In Fig.~\ref{fig:atten}, we show the power attenuation according to the distance in the antenna-based channel environment.
This figure plots the receive power for the given transmit power, from which we can derive the formula for the power attenuation according to the distance.
By regression analysis, the power attenuation is obtained as $h = \psi(d)= 0.01/d^{3.31}$, where $d$ is the distance in meter.

\subsection{Sensor Node Measurement}

\begin{figure}
    \centering
    \subfigure[I-V curve]{
        \label{fig:vigraph}\includegraphics[width=4.1cm, bb=1.0in 0.3in 9.5in 7.4in] {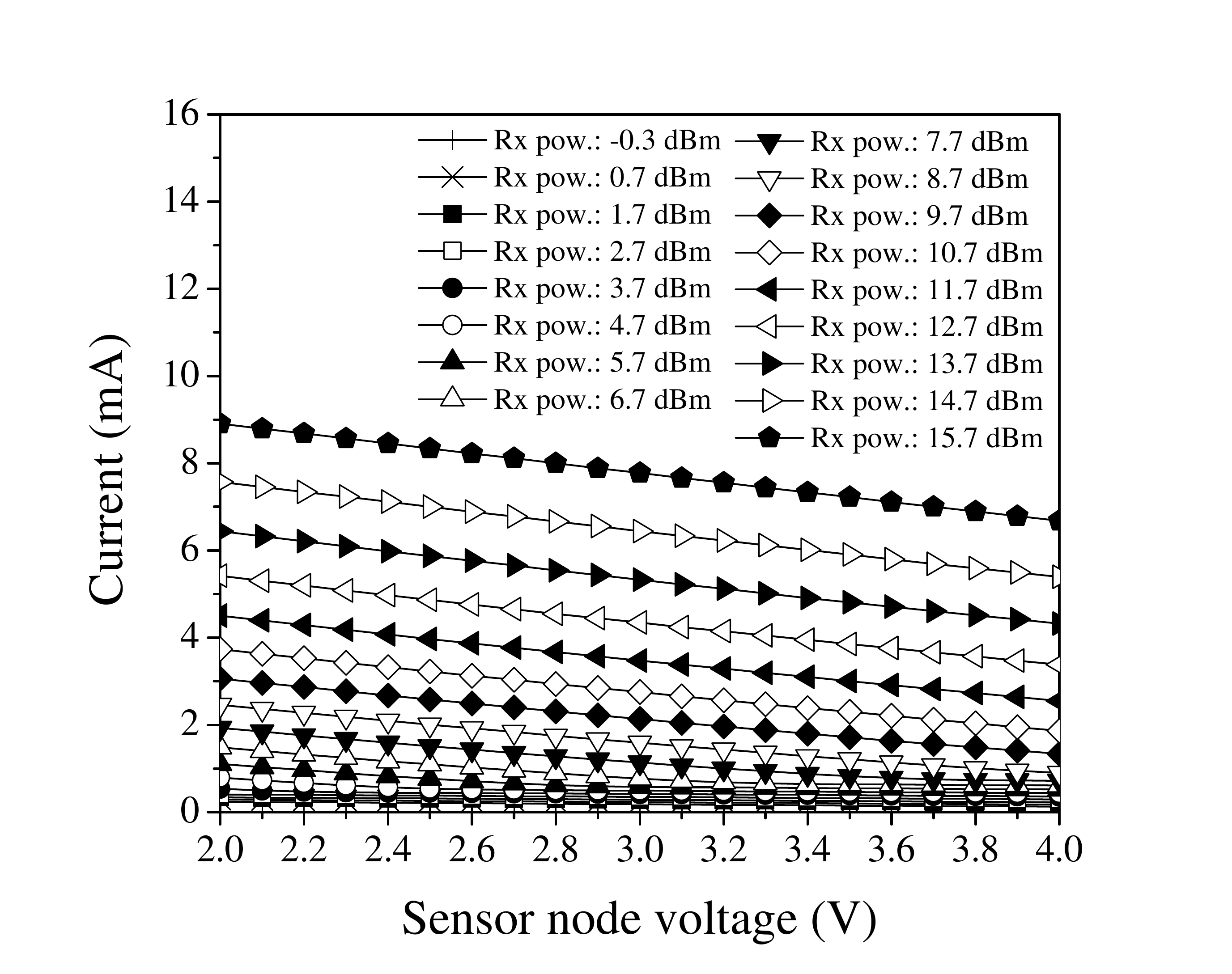}
        }~
    \subfigure[Efficiency]{
        \label{fig:harveff}\includegraphics[width=4.1cm, bb=1.0in 0.3in 9.5in 7.4in] {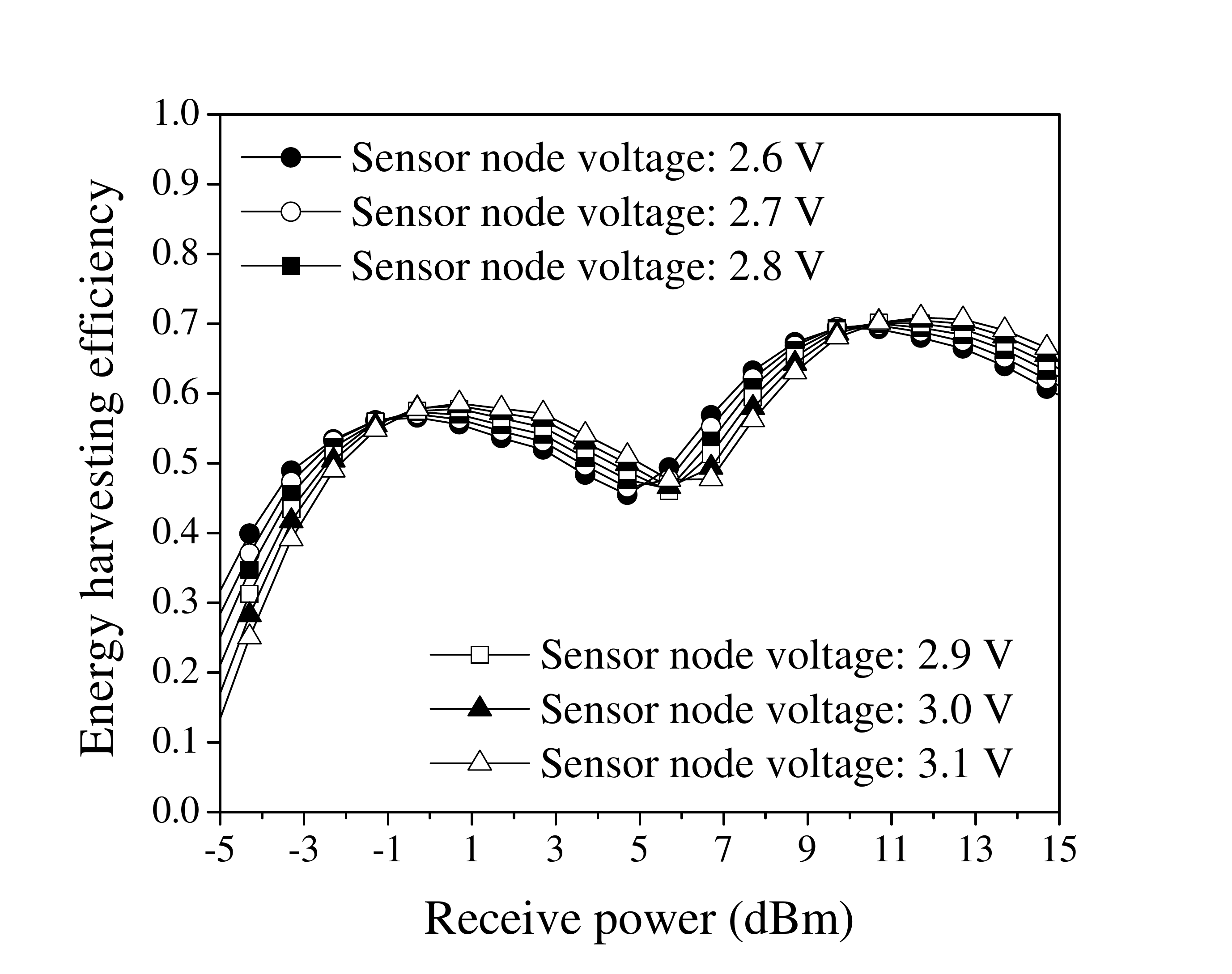}
        }
    \caption{I-V curve and efficiency of P1110 wireless energy harvester.}
    \label{fig:p1110}
\end{figure}

We have conducted tests on the wireless energy harvester (i.e., Powercast P1110) to obtain the I-V function $\rho$ and the wireless energy harvesting efficiency function $\eta_V$ for various receive power.
For the test, a signal generator is used to input a 920 MHz CW signal and an electronic load (i.e., Mayuno M9711) is utilized to measure the current according to various voltages.
Fig.~\ref{fig:vigraph} shows the measurement result of the I-V function of the wireless energy harvester.
The graphs in Fig.~\ref{fig:vigraph} are drawn for the receive power ranging from $-0.3$ dBm to $15.7$ dBm.
This figure defines the I-V function $\rho$ in \eqref{eq:ehiv} that can be used to calculate the current through the wireless energy harvester according to the sensor node voltage.
From the I-V function, we calculate the wireless energy harvesting efficiency as $\eta_V(V,p_\text{rx}) = V\rho(V,p_\text{rx})/p_\text{rx}$.
The wireless energy harvesting efficiency graph is plotted in Fig.~\ref{fig:harveff} as a function of the receive power.

\begin{figure}
    \centering
    \subfigure[Leakage of supercapacitor]{
        \label{fig:leakage}\includegraphics[width=4.1cm, bb=0.7in 0.3in 9.5in 7.4in] {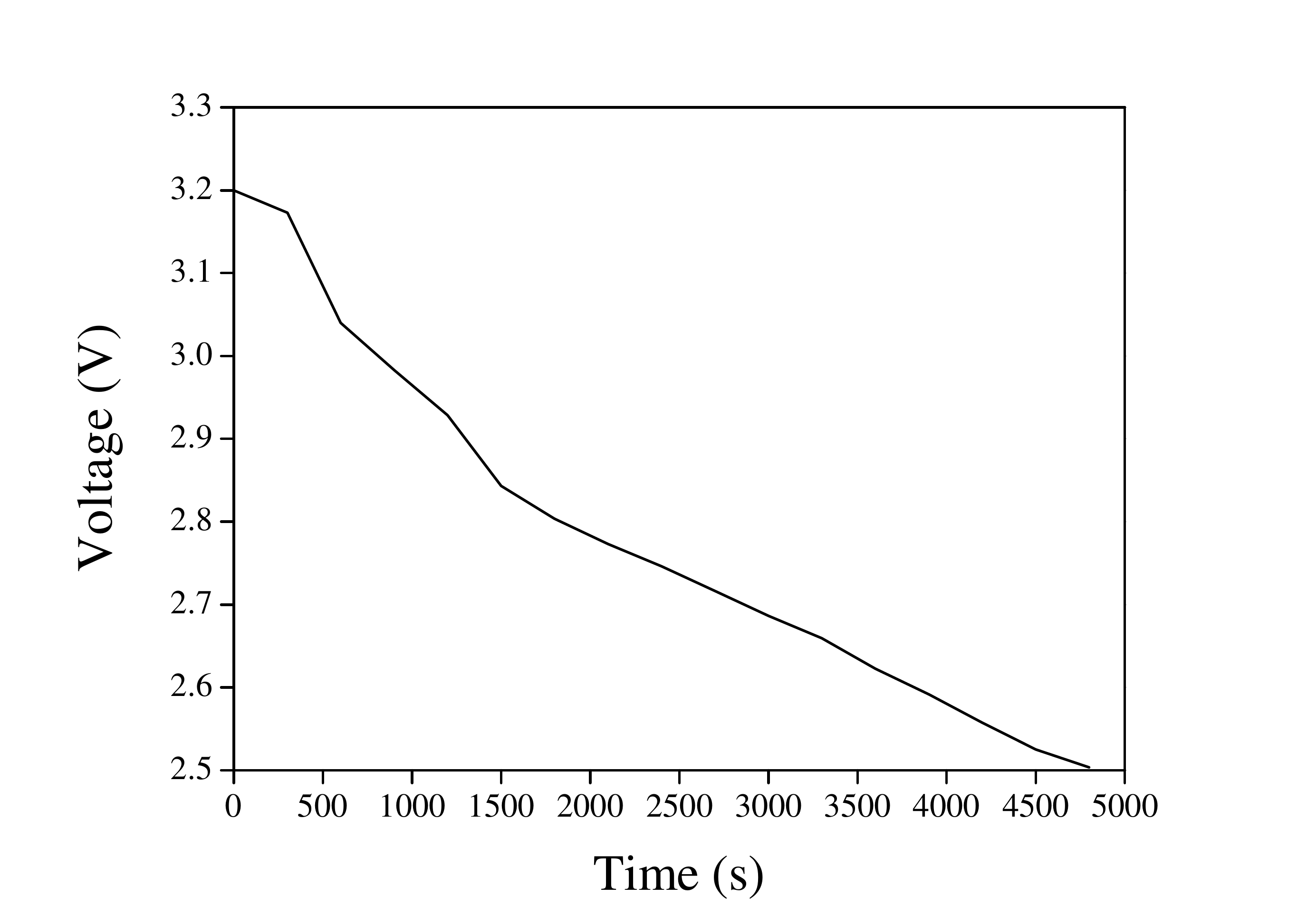}
        }~~
    \subfigure[Sensor module power consumption]{
        \label{fig:modepwr}\includegraphics[width=4.1cm, bb=0.7in 0.3in 9.5in 7.4in] {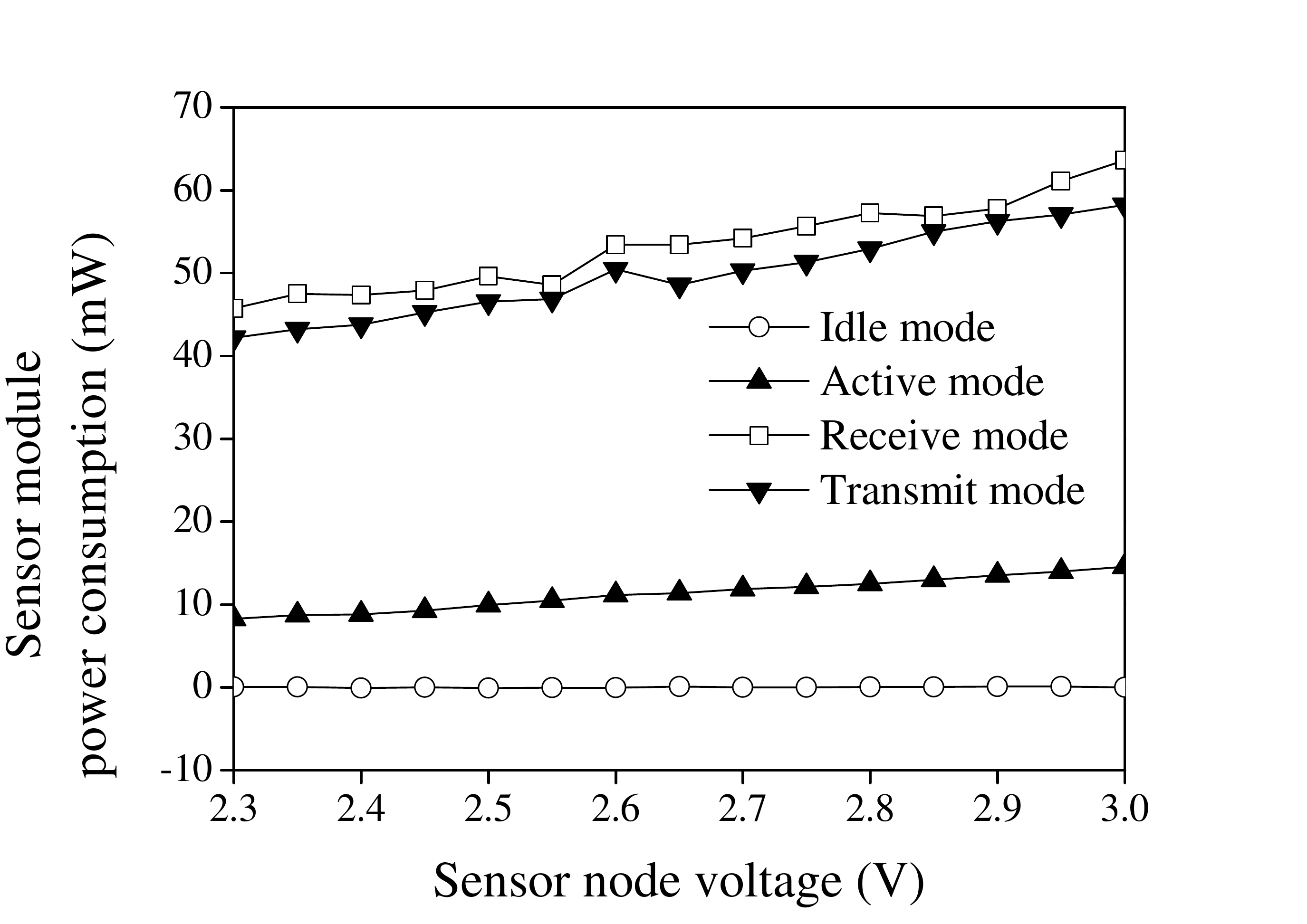}
        }
    \caption{Leakage of supercapacitor and sensor module power consumption.}
    \label{fig:leakagemodepwr}
\end{figure}

The parameters of the supercapacitor are obtained in the testbed as well.
The testbed uses a Samxon DDL series supercapacitor with the capacitance of $C=0.1$ F.
In Fig.~\ref{fig:leakage}, a leakage test is conducted by letting the supercapacitor discharge by itself.
The ideal capacitor and the leakage resistor in Fig.~\ref{fig:model} form an RC circuit when disconnected from other parts.
The voltage of an RC circuit evolves according to the equation $V(t) = V(0)\cdot \exp(-t/(R_\text{leak}C))$,
where $V(t)$ is the voltage at time $t$ and $V(0)$ is the initial voltage at time $t=0$.
By fitting this equation into the graph in Fig.~\ref{fig:leakage}, the leakage resistance is obtained as $R_\text{leak} = 196$ k$\Omega$.

In the testbed, the Zolertia Z1 mote is used as a sensor module.
We test the load characteristics of the Z1 mote for different modes of the sensor module.
In Fig.~\ref{fig:modepwr}, the power consumption of the Z1 mote is measured according to the voltage.
In the idle mode, very small power consumption is observed because of the leakage current of the ICs on the sensor module.
In the active mode, only the MCU (i.e., MSP430), which acts as a constant resistance load, is activated.
From Fig.~\ref{fig:modepwr}, the load resistance of the MSP430 is calculated to be 0.626 k$\Omega$.
In the receive and the transmit mode, the RF transceiver (i.e., CC2420) receives and transmits while the MSP430 is activated.
The power consumption of the CC2420 is obtained by subtracting out the power consumption of the MSP430 from the measured power consumption.
The CC2420 acts as a constant current load \cite{cc2420}, and the current consumption of the CC2420 is calculated to be 15.87 mA in the receive mode and 14.55 mA in the transmit mode.
These results are summarized in Table \ref{tab:chipmode}.
From this table, the resistance of the constant resistance load (CRL) (i.e., $\gamma(m)$) and the current of the constant current load (CCL) (i.e., $\zeta(m)$) are decided for each mode $m$.

\begin{table}
\caption{Measured parameters of the constant resistance and constant current loads for each mode of the Z1 mote.}
\centering
\begin{tabular}{|p{2.5em}|p{4em}|p{4em}|p{5em}|p{5em}|}\hline
\textbf{Load Type} & \textbf{Idle} & \textbf{Active} & \textbf{Receive} & \textbf{Transmit}\\
\hline\hline
CRL & $\infty$ & 0.626 k$\Omega$ & 0.626 k$\Omega$ & 0.626 k$\Omega$\\\hline
CCL & 0.035 mA & 0.035 mA & 15.87 mA & 14.55 mA \\\hline
\end{tabular}
\label{tab:chipmode}
\end{table}

\section{Energy Management Scheme for Energy Neutral Operation}\label{section:energy management}

\subsection{Description of Energy Management Scheme}

\begin{figure}
	\centering
    \includegraphics[width=6cm, bb=1.5in 3.4in 6.4in 7.3in]{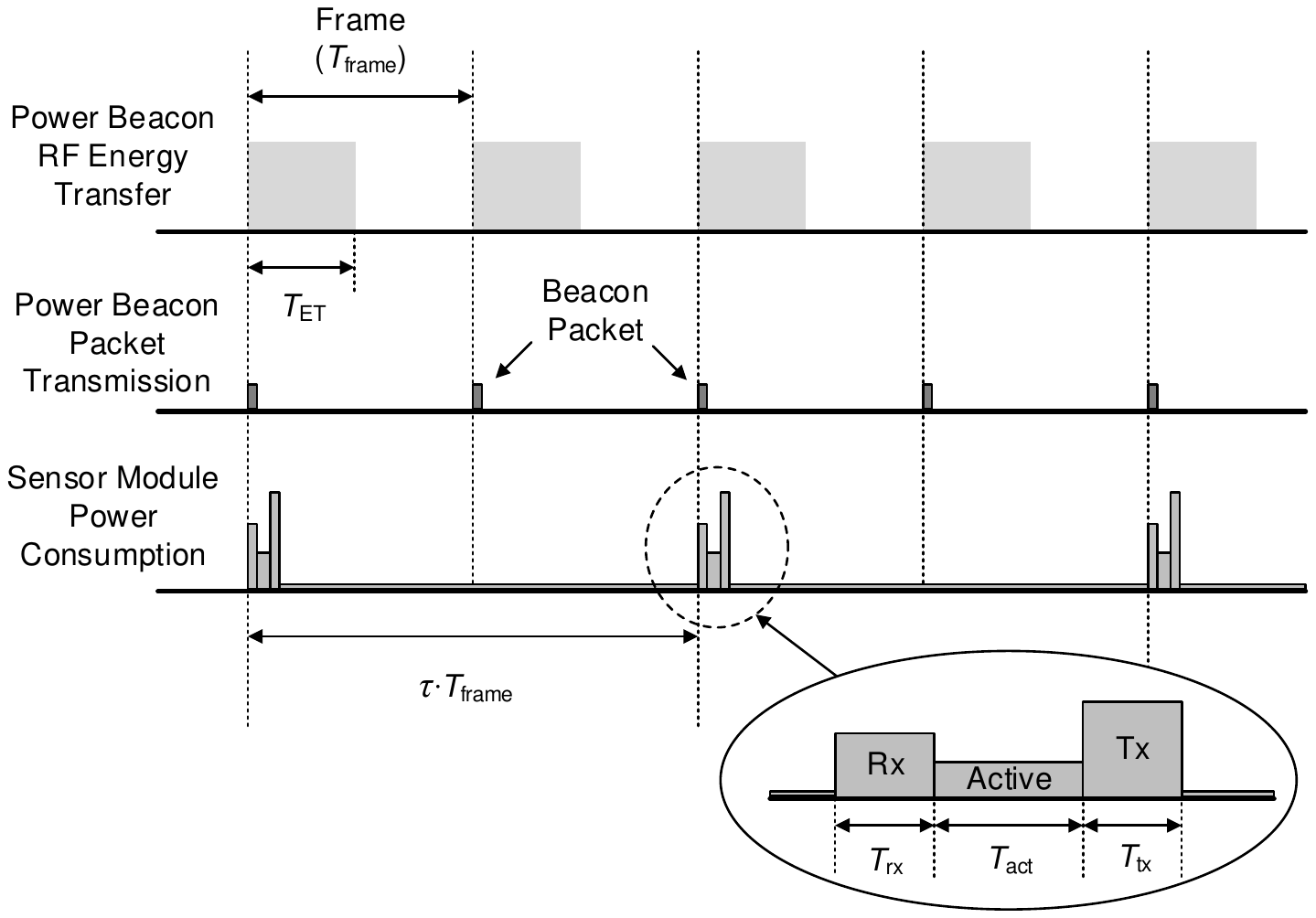}
    \caption{Timing diagram of the proposed energy management scheme.}
    \label{fig:energymanagement}
\end{figure}

In this section, we propose an energy management scheme for the WPSN.
The targets of the proposed energy management scheme are threefold.
First, the energy management scheme minimizes the amplifier power consumption in the power beacon.
Second, the energy management scheme guarantees that the power beacon receives sensor measurement reports from the sensor node as frequently as required by a sensor application.
Third, the energy management scheme maintains the stored energy $E$ over the minimum stored energy $E_\text{min}$ for assuring continuous operation of the sensor node.
These three goals are generally conflicting with each other.
Therefore, it is required to carefully design the energy management scheme so that these three goals are optimally balanced.

In the proposed scheme, both the power beacon and the sensor node perform duty cycling.
The timing diagram of the proposed energy management scheme is illustrated in Fig.~\ref{fig:energymanagement}.
The basic time unit for all operations of the energy management scheme is a frame, whose length is denoted by $T_\text{frame}$.
Each frame is indexed by $k$.
The power beacon performs duty cycling on a frame-by-frame basis.
During each frame, the power beacon turns on the amplifier for $T_\text{ET}$ and turns off the amplifier for the rest of the frame.
Then, the amplifier duty cycle, denoted by $\alpha$, is given by
\begin{align}
\alpha = T_\text{ET}/T_\text{frame}.
\end{align}

When the amplifier is turned on, the transmit power is set to $p_\text{tx} = \Upsilon$.
Hereafter, we will call $\Upsilon$ energy transfer power.
The energy transfer power satisfies $0 < \Upsilon \le \Upsilon_\text{max}$, where $\Upsilon_\text{max}$ denotes the maximum energy transfer power.
From \eqref{eq:ampcons}, the amplifier power consumption is given by
\begin{align}\label{eq:ampcons2}
p_\text{cons} = \Upsilon/\theta(\Upsilon),
\end{align}
when the amplifier is turned on.
On the other hand, when the amplifier is turned off, the transmit power and the amplifier power consumption are both zero.
Let $\alpha(k)$ and $\Upsilon(k)$ denote the amplifier duty cycle and the energy transfer power in frame $k$.

The sensor node performs duty cycling that periodically wakes up the sensor node for a short time and then puts the sensor node into the idle mode to minimize the sensor module power consumption.
A frame is a basic time unit of the duty cycling of the sensor node.
For each frame, the sensor node can be in the awake state or in the sleep state.
If the sensor node is in the awake state during a frame, the sensor node goes through the receive, active, and transmit modes in sequence for $T_\text{rx}$, $T_\text{act}$, and $T_\text{tx}$, respectively, and it goes into the idle mode for the rest of the frame during $T_\text{idle}=T_\text{frame}-T_\text{rx}-T_\text{act}-T_\text{tx}$, as shown in Fig.~\ref{fig:energymanagement}.
On the other hand, if the sensor node is in the sleep state during a frame, it stays in the idle mode for the entire frame.
Let $a(k)$ denote an awake indicator that is 1 if the sensor node is in the awake state in frame $k$; and is 0, otherwise.

Suppose that the sensor node is in the awake state in frame $k$ (i.e., $a(k)=1$).
Then, the sensor node sets up the wake-up timer before it goes into the idle mode so that it can wake up again after the wake-up interval.
The wake-up interval in frame $k$ is denoted by $\tau(k)$.
Let $\sigma(k)$ denote the remaining number of frames until the sensor node wakes up again.
In frame $k$ such that $\sigma(k)>0$, the sensor node is in the sleep state (i.e., $a(k)=0$).
On the other hand, if $\sigma(k)=0$, the sensor node is in the awake state in frame $k$ (i.e., $a(k)=1$).
In each frame $k$ such that $\sigma(k)>0$, $\sigma(k)$ counts down by one (i.e., $\sigma(k+1)=\sigma(k)-1$ if $\sigma(k)>0$).
If $\sigma(k)$ becomes zero, the sensor node wakes up (i.e., $a(k)=1$) and it sets $\sigma(k+1)$ to $(\tau(k)-1)$ so that it wakes up again after $\tau(k)$ frames (i.e., $\sigma(k+1)=\tau(k)-1$ if $\sigma(k)=0$).

At the start of each frame, the power beacon sends a beacon packet that contains a control command to the sensor node.
When the sensor node wakes up at the start of a frame, the sensor node is in the receive mode for time duration $T_\text{rx}$ to receive a beacon packet from the power beacon.
The sensor node wakes up slightly before the start of a frame to make sure the sensor node safely receives a beacon packet.
After successful reception of a beacon packet, the sensor node turns off the RF transceiver and the mode of the sensor node is changed to the active mode.

In the active mode, the sensor node performs sensor measurements and computation for time duration $T_\text{act}$.
The sensor measurements include the receive power, the stored energy, and other sensor measurements (e.g., temperature).
During the time that the sensor node measures the receive power, the amplifier in the power beacon should be turned on.
Therefore, the energy management scheme makes the amplifier duty cycle larger than a predefined minimum amplifier duty cycle $\alpha_\text{min}$ so that the amplifier is turned on while the sensor node is in the active mode.
The receive power measurement in frame $k$ is denoted by $\overline{S}(k)$.
In addition, the sensor node measures the sensor node voltage and calculates the stored energy from \eqref{eq:sten}.
The stored energy measurement in frame $k$ is denoted by $\overline{E}(k)$.
The sensor node prepares the sensor measurement report that includes $\overline{S}(k)$, $\overline{E}(k)$, and other sensor measurements.
Then, the sensor node sends a data packet containing the sensor measurement report in the transmit mode for time duration $T_\text{tx}$.
After the packet transmission is over, the sensor node is put into the idle mode.
The sensor node wakes up again after the wake-up interval.

The proposed energy management scheme has three control parameters: the amplifier duty cycle, the energy transfer power, and the wake-up interval.
The energy management scheme can dynamically decide the amplifier duty cycle, the energy transfer power, and the wake-up interval for each frame in which the sensor node is in the awake state.
This decision is based on the receive power measurement $\overline{S}(k)$ and the stored energy measurement $\overline{E}(k)$ included in the sensor measurement report from the sensor node.
The controller sends the decided wake-up interval to the sensor node by enclosing it in the control command in the beacon packet, and the sensor node adjusts the wake-up interval according to the received control command.

\subsection{Discrete-Time Stored Energy Evolution Model of Energy Management Scheme}\label{section:evolutionmodel}

In this subsection, we derive the discrete-time stored energy evolution model when the proposed energy management scheme is applied.
The stored energy at the start of frame $k$ is denoted by $E(k)$.
From the continuous-time energy transition function in \eqref{eq:evol2}, the discrete-time energy evolution formula is obtained as
\begin{align}
\begin{split}
E(k+1) &= \min\big\{E(k) + \mbox{$\int^{t(k+1)}_{t(k)}$} \big(\phi(E^{(t)},p_\text{rx}^{(t)})\\
&\quad\quad - \xi_{m^{(t)}}(E^{(t)}) - \xi_\text{leak}(E^{(t)})\big)\mathrm{d}t, E_\text{max}\big\},
\end{split}
\end{align}
where $t(k)$ is the time at the start of frame $k$, and $E^{(t)}$, $p_\text{rx}^{(t)}$, and $m^{(t)}$ are the stored energy, the receive power, and the mode of the sensor module at time $t$, respectively.

To simplify the discrete-time energy evolution formula, we assume that the variation of $E^{(t)}$ during one frame does not much affect the value of $\phi(E^{(t)},p_\text{rx}^{(t)})$, $\xi_{m^{(t)}}(E^{(t)})$, and $\xi_\text{leak}(E^{(t)})$.
Then, we can assume that $\phi(E(k),p_\text{rx}^{(t)})=\phi(E^{(t)},p_\text{rx}^{(t)})$, $\xi_{m^{(t)}}(E(k))=\xi_{m^{(t)}}(E^{(t)})$, and $\xi_\text{leak}(E(k))=\xi_\text{leak}(E^{(t)})$ for $t(k) \le t < t(k+1)$.
This assumption is valid when the capacitance of the supercapacitor (i.e., $C$) is sufficiently large.
Since $V=\sqrt{2E/C}$, the large capacitance makes the sensor node voltage changes slowly, and we can consider that the sensor node voltage does not change during a frame.
Then, the above assumption holds since the harvested power, the sensor module power consumption, and the supercapacitor leakage power are all functions of the sensor node voltage.

Under this assumption, the discrete-time energy evolution is given by
\begin{align}\label{eq:discenevol}
\begin{split}
E(k+1)& = \min\{E(k)+\Phi(E(k),\alpha(k),\Upsilon(k),h)\\
&\qquad\quad-\Theta(E(k),a(k)),\ E_\text{max}\},
\end{split}
\end{align}
where $\Phi(E,\alpha,\Upsilon,h)$ is the harvested energy and $\Theta(E,a)$ is the consumed energy during one frame when the stored energy is $E$, the amplifier duty cycle is $\alpha$, the energy transfer power is $\Upsilon$, the power attenuation is $h$, and the awake indicator is $a$.
In \eqref{eq:discenevol}, the harvested energy is defined as
\begin{align}\label{eq:harven}
\Phi(E,\alpha,\Upsilon,h) = \alpha\cdot\eta_E(E,h\Upsilon)\cdot h\Upsilon\cdot T_\text{frame},
\end{align}
and the consumed energy is defined as
\begin{align}\label{eq:consen}
\begin{split}
&\Theta(E,a)\\
&=a\cdot\big(\mbox{$\sum_{m\in\{\text{rx,act,tx,idle}\}}$}\xi_m(E)\cdot T_m + \xi_\text{leak}(E)\cdot T_\text{frame}\big)\\
&\quad+ (1-a)\cdot(\xi_\text{idle}(E) + \xi_\text{leak}(E)))\cdot T_\text{frame}\\
&= \varphi(E) + \delta(E)\cdot a,
\end{split}
\end{align}
where $\varphi(E) = (\xi_\text{idle}(E) + \xi_\text{leak}(E))\cdot T_\text{frame}$ and $\delta(E) = \mbox{$\sum_{m\in\{\text{rx,act,tx}\}}$}(\xi_m(E)-\xi_\text{idle}(E))\cdot T_m$.

The power amplifier at the power beacon consumes DC power according to \eqref{eq:ampcons2} only when it is turned on.
Therefore, the average amplifier power consumption is a function of the amplifier duty cycle and the energy transfer power.
The average amplifier power consumption is defined as
\begin{align}\label{eq:omega}
\Omega(\alpha,\Upsilon) = \alpha\cdot (\Upsilon/\theta(\Upsilon)).
\end{align}

\subsection{Optimal Energy Transfer Strategy}\label{section:optstrategy}

The proposed energy management scheme aims to minimize the average amplifier power consumption $\Omega(\alpha,\Upsilon)$ while maintaining the wake-up interval $\tau$ to the target wake-up interval $\tau_\text{tgt}$.
In doing so, the energy management scheme stabilizes the stored energy $E$ at the target stored energy $E_\text{tgt}$, which is higher than the minimum stored energy $E_\text{min}$.

Let us define the awake frame ratio, denoted by $r$, as the ratio of the frames in which the sensor node is in the awake state.
Since the sensor node is in the awake state in one frame out of $\tau$ frames, the awake frame ratio is $r = 1/\tau$.
Then, the average consumed energy is given by
\begin{align}
\begin{split}
Q(E,r) &= r\cdot \Theta(E,1)- (1-r)\cdot \Theta(E,0)\\
&= \varphi(E) + \delta(E)\cdot r.
\end{split}
\end{align}

While the stored energy and the wake-up interval are maintained to $E_\text{tgt}$ and $\tau_\text{tgt}$, respectively, the energy management scheme finds the optimal amplifier duty cycle $\alpha^*$ and the optimal energy transfer power $\Upsilon^*$ of the following optimization problem:
\begin{align}
&&&\text{minimize} & &\Omega(\alpha,\Upsilon)&&&&\label{eq:opttarget}\\
&&&\text{subject to} & &\Phi(E_\text{tgt},\alpha,\Upsilon,h) \ge Q(E_\text{tgt},r_\text{tgt}), &&&&\label{eq:optconst}
\end{align}
where $r_\text{tgt}=1/\tau_\text{tgt}$, $\alpha_\text{min}\le \alpha\le 1$, and $0< \Upsilon \le \Upsilon_\text{max}$.
Note that the optimization problem \eqref{eq:opttarget} and \eqref{eq:optconst} is not convex.

The Lagrangian of the optimization problem \eqref{eq:opttarget} and \eqref{eq:optconst} is given by
\begin{align}\label{eq:lagrangian}
\begin{split}
&L(\alpha,\Upsilon,\mu)\\
&= \Omega(\alpha,\Upsilon) - \mu\Phi(E_\text{tgt},\alpha,\Upsilon,h)+\mu Q(E_\text{tgt},r_\text{tgt})\\
&=\alpha\Upsilon (\theta(\Upsilon)^{-1} - \mu\cdot \eta_E(E_\text{tgt},h\Upsilon) h)+\mu Q(E_\text{tgt},r_\text{tgt}),
\end{split}
\end{align}
where $\mu \ge 0$ is the Lagrange multiplier.
The dual function is $g(\mu) = \min_{\alpha,\Upsilon} L(\alpha,\Upsilon,\mu)$.
It is known that the dual function is always smaller than or equal to the optimal value of \eqref{eq:opttarget} and \eqref{eq:optconst}, that is, $g(\mu)\le \Omega(\alpha^*,\Upsilon^*)$ for $\mu \ge 0$.

Let $\widehat{\mu}$ denote the minimum $\mu$ that satisfies $\theta(\Upsilon)^{-1} - \mu\cdot \eta_E(E_\text{tgt},h\Upsilon) h = 0$ for some $\Upsilon$ over $0< \Upsilon \le \Upsilon_\text{max}$.
For such $\widehat{\mu}$, we define $\widehat{\Upsilon}$ that satisfies $\theta(\widehat{\Upsilon})^{-1} - \widehat{\mu}\cdot \eta_E(E_\text{tgt},h\widehat{\Upsilon}) h = 0$.
Henceforth, we will call $\widehat{\Upsilon}$ maximum efficiency energy transfer power.
In addition, we define the maximum efficiency receive power as $\widehat{S} = h\widehat{\Upsilon}$ and the maximum efficiency harvested power as $\widehat{H} = \eta_E(E_\text{tgt},\widehat{S})\cdot \widehat{S}$.

The optimal solutions of \eqref{eq:opttarget} and \eqref{eq:optconst} satisfy the following theorem.
\begin{theorem}\label{theorem:optimal}
If $Q(E_\text{tgt},r_\text{tgt}) \le \widehat{H}$, the optimal solutions are $\alpha^* =  Q(E_\text{tgt},r_\text{tgt})/\widehat{H}$ and $\Upsilon^* = \widehat{\Upsilon}$.
\end{theorem}
\begin{proof}
See Appendix \ref{proof:optimal}.
\end{proof}

The maximum efficiency energy transfer power $\widehat{\Upsilon}$ mainly depends on two efficiency functions, the wireless energy harvesting efficiency function $\eta_E$ and the PAE function $\theta$.
Theorem \ref{theorem:optimal} states that, if $Q(E_\text{tgt},r_\text{tgt})$ is less than or equal to a threshold $\widehat{H}$, it is efficient to send the RF energy with the maximum efficiency energy transfer power $\widehat{\Upsilon}$ and to set the amplifier duty cycle less than one so that the amplifier power consumption is minimized while the average consumed energy is supported.

\subsection{Adaptive Energy Management Algorithm}\label{section:emalg}

In this subsection, we propose an adaptive energy management algorithm that controls the amplifier duty cycle $\alpha(k)$, the energy transfer power $\Upsilon(k)$, and the wake-up interval $\tau(k)$.
The energy management algorithm updates $\alpha(k)$, $\Upsilon(k)$, and $\tau(k)$ only in the frame in which the sensor node is in the awake state (i.e., frame $k$ such that $a(k)=1$).
The start of frame $k$ such that $a(k)=1$ is defined as an energy management epoch, and the $i$th energy management epoch will be called epoch $i$.
Let $E_i$, $\alpha_i$, $\Upsilon_i$, $\tau_i$, $\overline{E}_i$, and $\overline{S}_i$ denote the stored energy, the amplifier duty cycle, the energy transfer power, the wake-up interval, the stored energy measurement, and the receive power measurement at epoch $i$, respectively.
Then, the stored energy evolves over epochs according to the following formula.
\begin{align}\label{eq:epochenevol}
\begin{split}
E_{i+1}& = \min\{E_i+(\Phi(E_i,\alpha_i,\Upsilon_i,h)\\
&\qquad\qquad\qquad\quad-Q(E_i,1/\tau_i))\cdot\tau_i,\ E_\text{max}\},
\end{split}
\end{align}

The power beacon controls $\alpha_i$, $\Upsilon_i$, and $\tau_i$ based on the receive power measurement $\overline{S}_i$ and the stored energy measurement $\overline{E}_i$ enclosed in the sensor measurement report from the sensor node.
Note that this algorithm treats the wake-up interval $\tau_i$ as a real number rather than as an integer.
Therefore, the wake-up interval obtained by this algorithm should be rounded off when it is actually applied.
The first priority of the energy management algorithm is to maintain the stored energy $E_i$ to the target stored energy $E_\text{tgt}$.
The second priority is to keep the wake-up interval $\tau_i$ to the target wake-up interval $\tau_\text{tgt}$.
As long as the above two targets are satisfied, the algorithm tries to minimize the average amplifier power consumption according to Theorem \ref{theorem:optimal}.

According to the average consumed energy, $Q(E_\text{tgt},r_\text{tgt})$, required for achieving $E_\text{tgt}$ and $\tau_\text{tgt}$, we can consider the following three cases.
In Case I, it is satisfied that $Q(E_\text{tgt},r_\text{tgt}) \le \widehat{H}$.
This case corresponds to Theorem \ref{theorem:optimal}.
This is the most common case if the wireless energy harvesting module is designed in such a way that the maximum energy harvesting efficiency is achieved in the desired operating condition.
To achieve the optimality in Case I, $\alpha_i$ should be controlled to maintain the target stored energy while $\Upsilon_i$ is set to $\widehat{\Upsilon}$, according to Theorem \ref{theorem:optimal}.
In Case II, we have $\widehat{H} < Q(E_\text{tgt},r_\text{tgt}) \le \eta_E(E_\text{tgt},h\Upsilon_\text{max})\cdot h\Upsilon_\text{max}$.
In this case, $Q(E_\text{tgt},r_\text{tgt})$ cannot be supported with $\Upsilon_i = \widehat{\Upsilon}$ and $\alpha_i = 1$.
Therefore, the proposed algorithm controls $\Upsilon_i$ to a value higher than $\widehat{\Upsilon}$ while it sets $\alpha_i=1$.
In Case III, we have $Q(E_\text{tgt},r_\text{tgt}) > \eta_E(E_\text{tgt},h\Upsilon_\text{max})\cdot h\Upsilon_\text{max}$.
In this case, $Q(E_\text{tgt},r_\text{tgt})$ cannot be supported even with $\Upsilon_i = \Upsilon_\text{max}$ and $\alpha_i = 1$.
The only way to maintain the stored energy is to adjust $\tau_i$ to a value higher than $\tau_\text{tgt}$ to reduce the average consumed energy.

The proposed algorithm controls $\alpha_i$ for Case I, $\Upsilon_i$ for Case II, and $\tau_i$ for Case III.
Since one parameter is controlled at a time in all three cases, we can define one control variable $x_i$ that is mapped to the tuple of three parameters $(\alpha_i,\Upsilon_i,\tau_i)$.
The mapping function is ${\boldsymbol \omega(x)} = (\omega_\alpha(x),\omega_\Upsilon(x),\omega_\tau(x))$, which is defined as
\begin{align}\label{eq:omegax}
\begin{split}
{\boldsymbol \omega(x)}
=\begin{cases}
(x+\alpha_\text{min},\widehat{\Upsilon},\tau_\text{tgt}),&\text{if }0\le x\le \kappa_1\\
(1,\beta_\Upsilon(x-\kappa_1)+\widehat{\Upsilon},\tau_\text{tgt}),&\text{if }\kappa_1< x\le \kappa_2\\
(1,\Upsilon_\text{max},\beta_\tau(x-\kappa_2)+\tau_\text{tgt}),&\text{if }x> \kappa_2,
\end{cases}
\end{split}
\end{align}
where $\kappa_1 = 1-\alpha_\text{min}$, $\kappa_2 = (\Upsilon_\text{max}-\widehat{\Upsilon})/\beta_\Upsilon+1-\alpha_\text{min}$, and $\beta_\Upsilon$ and $\beta_\tau$ are positive constants.
The stored energy evolution in \eqref{eq:epochenevol} can be rewritten as a function of $x$ as follows.
\begin{align}\label{eq:epochenevol2}
\begin{split}
E_{i+1}& = \min\{E_i+\Delta(E_i,x_i),E_\text{max}\},
\end{split}
\end{align}
where $\Delta(E,x)$ is defined as
\begin{align}\label{eq:delta}
\begin{split}
&\Delta(E,x) \\
&= (\Phi(E,\omega_\alpha(x),\omega_\Upsilon(x),h)-Q(E,1/\omega_\tau(x)))\cdot\omega_\tau(x)\\
&= (\omega_\alpha(x)\eta_E(E,h\omega_\Upsilon(x))h\omega_\Upsilon(x)\\
&\qquad\qquad-\varphi(E)-\delta(E)/\omega_\tau(x))\cdot\omega_\tau(x).
\end{split}
\end{align}
From \eqref{eq:delta}, we can see that $\Delta(E,x)$ is a non-decreasing function of $x\ge 0$ for all $E_\text{min}\le E\le E_\text{max}$.

The stored energy evolution in \eqref{eq:epochenevol2} is an integrating process.
If the plant is an integrating process, the proportional-integral (PI) controller can be used to control $x_i$ for maintaining $E_i$ close to $E_\text{tgt}$.
The PI controller is
\begin{align}\label{eq:pic}
x_i = C_P (E_\text{tgt}-\overline{E}_i) + C_I \mbox{$\sum_{j=1}^i$} (E_\text{tgt}-\overline{E}_j),
\end{align}
where $C_P$ and $C_I$ are the coefficients for the proportional and integral terms, respectively.
After calculating \eqref{eq:pic}, the algorithm can derive $\alpha_i = \omega_\alpha(x_i)$, $\Upsilon_i = \omega_\Upsilon(x_i)$, and $\tau_i = \omega_\tau(x_i)$.

In calculating \eqref{eq:omegax}, the algorithm should know the maximum efficiency transfer power $\widehat{\Upsilon}$.
However, it is difficult to find the exact $\widehat{\Upsilon}$ since $\widehat{\Upsilon}$ depends on nonlinear functions as well as the power attenuation.
Therefore, in our experiment, we assume that the maximum efficiency receive power $\widehat{S}$ has a constant value $S_\text{tgt}$.
The power attenuation is obtained as $h = \overline{S}_i/\Upsilon_i$ based on the receive power measurement $\overline{S}_i$.
Then, the maximum efficiency energy transfer power can be calculated as $\widehat{\Upsilon}=S_\text{tgt}/h$.

\section{Experimental Results}\label{section:result}

In this section, we present experimental results on the WPSN testbed to show the performance of the proposed energy management scheme.
We have implemented the energy management scheme on the WPSN testbed as proposed in Section \ref{section:energy management}.
The length of a frame is $T_\text{frame} = 100$ ms, the maximum energy transfer power is $\Upsilon_\text{max} = 2.3$ W, and the maximum sensor node voltage is 3.0 V.
In this section, the power attenuation on the RF energy transfer channel will be given in a dB scale (i.e., $-10\log_{10} h$).

In Figs.~\ref{fig:ampdutycycling} and \ref{fig:ampharvcons}, we show the efficiency of the RF energy transfer when the amplifier duty cycling is applied.
Fig.~\ref{fig:ampduty} shows the average amplifier power consumption (i.e., $\Omega(\alpha,\Upsilon)=\alpha\cdot(\Upsilon/\theta(\Upsilon))$) on y-axis and the average transmit power (i.e., $\alpha\Upsilon$) on x-axis, according to the amplifier duty cycle $\alpha$ and the energy transfer power $\Upsilon$.
In this figure, we can see that lower average amplifier power consumption is achieved for a given average transmit power when the amplifier duty cycling is used with fixed $\Upsilon$, which is because of the characteristics of the PAE.
In Fig.~\ref{fig:dutyharv}, we show the average harvested power (i.e., $\alpha\cdot\eta_E(E,h\Upsilon)\cdot h\Upsilon$) on y-axis and the average receive power (i.e., $\alpha\cdot h\Upsilon$) on x-axis, according to $\alpha$ and the receive power (i.e., $h\Upsilon$).
Since the wireless energy harvesting efficiency is maximized around 10 mW from Fig.~\ref{fig:harveff}, we have higher average harvested power when the receive power is fixed to 10 mW.
Fig.~\ref{fig:ampharvcons} shows the average amplifier power consumption for achieving a given average harvested power, which is affected by both the PAE and the wireless energy harvesting efficiency, when the power attenuations are 19.73 dB and 22.65 dB.
This figure clearly shows the advantage of the amplifier duty cycling.

\begin{figure}
    \centering
    \subfigure[Amplifier power consumption]{
        \label{fig:ampduty}\includegraphics[width=4.1cm, bb=0.7in 0.3in 9.5in 7.4in] {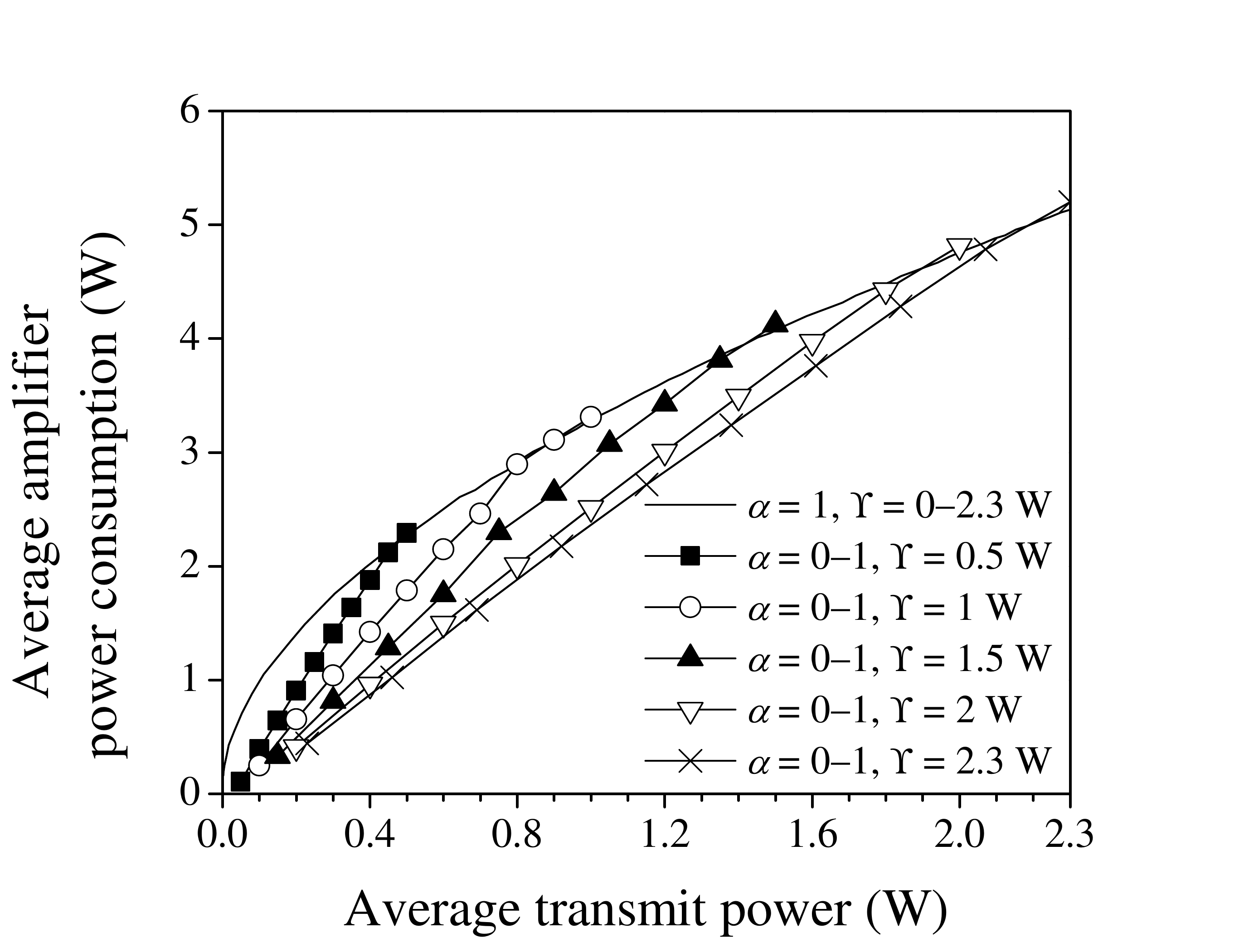}
        }~
    \subfigure[Harvested power]{
        \label{fig:dutyharv}\includegraphics[width=4.1cm, bb=0.7in 0.3in 9.5in 7.4in] {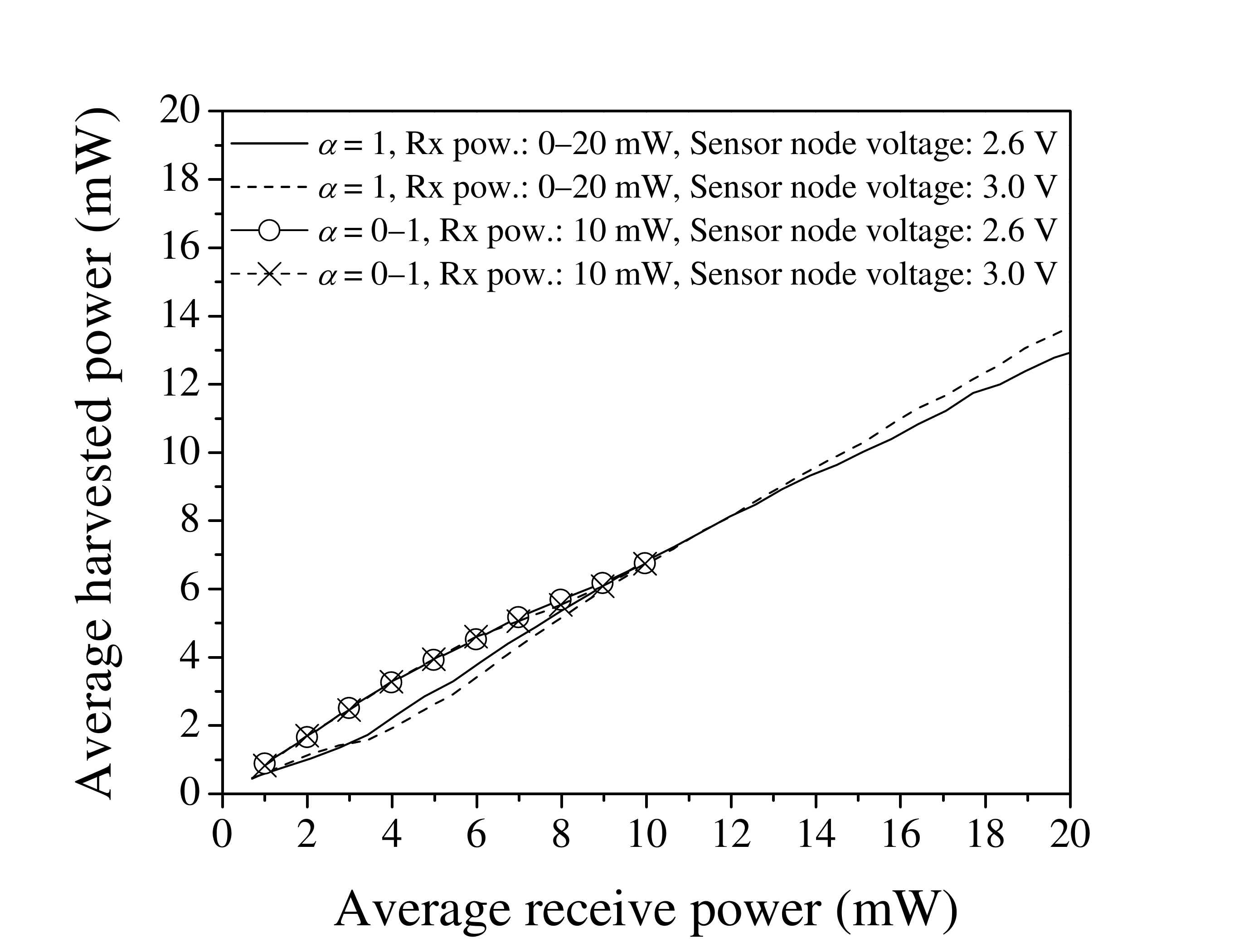}
        }
    \caption{Average amplifier power consumption and average harvested power when amplifier duty cycling is used.}
    \label{fig:ampdutycycling}
\end{figure}

\begin{figure}
	\centering
    \includegraphics[width=4.5cm, bb=0.8in 0.3in 9.6in 7.2in]{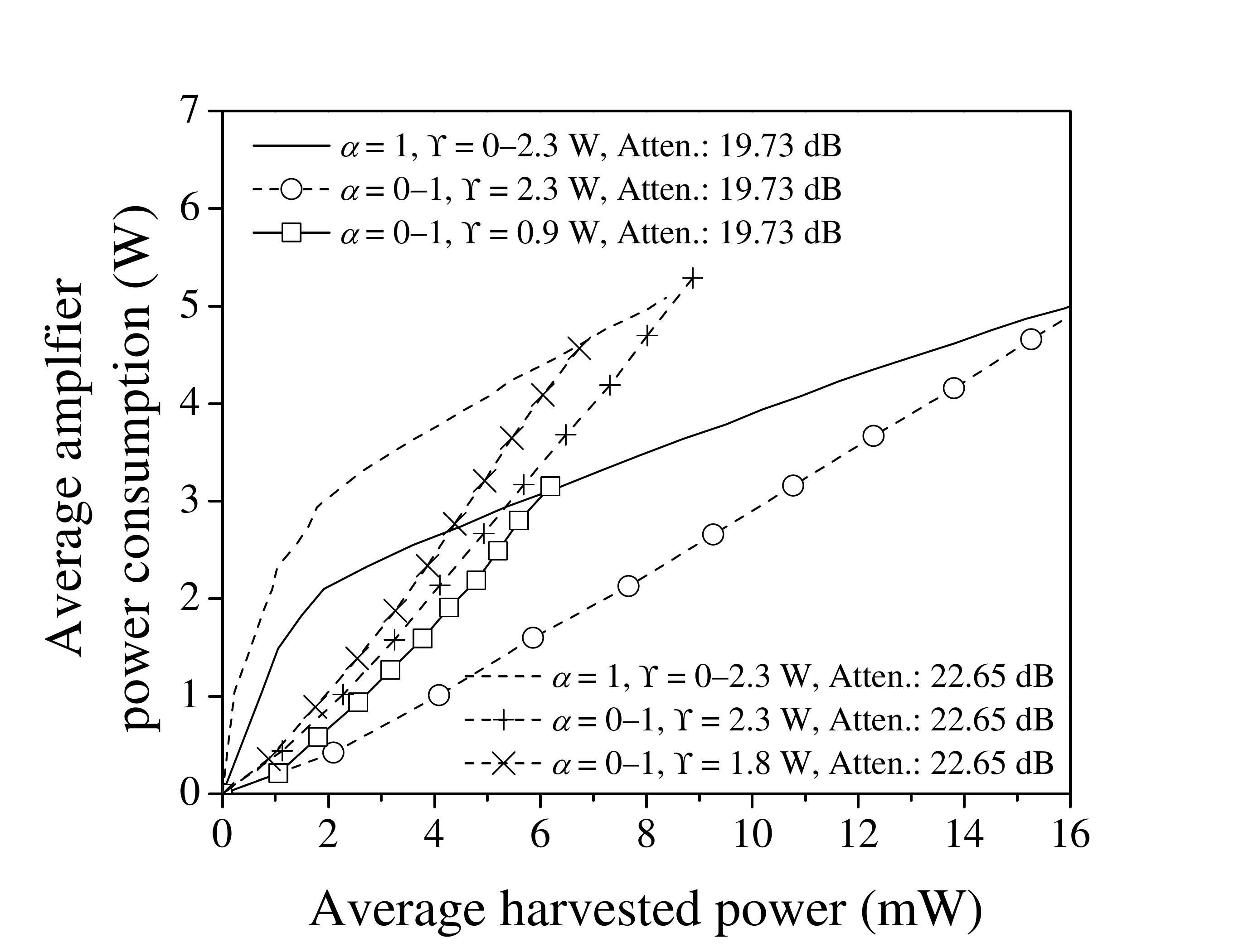}
    \caption{Average amplifier power consumption according to average harvested power when amplifier duty cycling is used.}
    \label{fig:ampharvcons}
\end{figure}

Fig.~\ref{fig:senscons} shows the sensor module power consumption over time when the wake-up interval is set to $\tau = 1$ and the sensor node voltage is 3.0 V.
In Fig.~\ref{fig:sensconslong}, we can see the sensor module wakes up every 100 ms similarly to the timing diagram in Fig.~\ref{fig:energymanagement}.
From Fig.~\ref{fig:sensconsshort}, we can derive the time duration of each mode such that $T_\text{rx} = 2.34$ ms, $T_\text{act} = 5.01$ ms, and $T_\text{tx} = 1.81$ ms.
In Fig.~\ref{fig:intpwr}, we show the average sensor module power consumption according to the wake-up interval.
The average sensor module power consumption is given by
$(\sum_{m\in \{\text{rx},\text{act},\text{tx},\text{idle}\}} \xi_m(E)\cdot T_m
 + \xi_\text{idle}(E) \cdot (\tau-1) T_\text{frame})/(\tau T_\text{frame})$.
In Fig.~\ref{fig:intpwr}, `expected' refers to the average sensor module power consumption calculated by this equation based on $T_\text{rx}$, $T_\text{act}$, $T_\text{tx}$, and Table \ref{tab:chipmode}.
In addition, we have measured the actual sensor node power consumption by measuring current going through the sensor module.
We can see that the expected sensor module power consumption well matches with the measured one.

\begin{figure}
    \centering
    \subfigure[Long time period]{
        \label{fig:sensconslong}\includegraphics[width=4.1cm, bb=0.7in 0.3in 9.5in 7.4in] {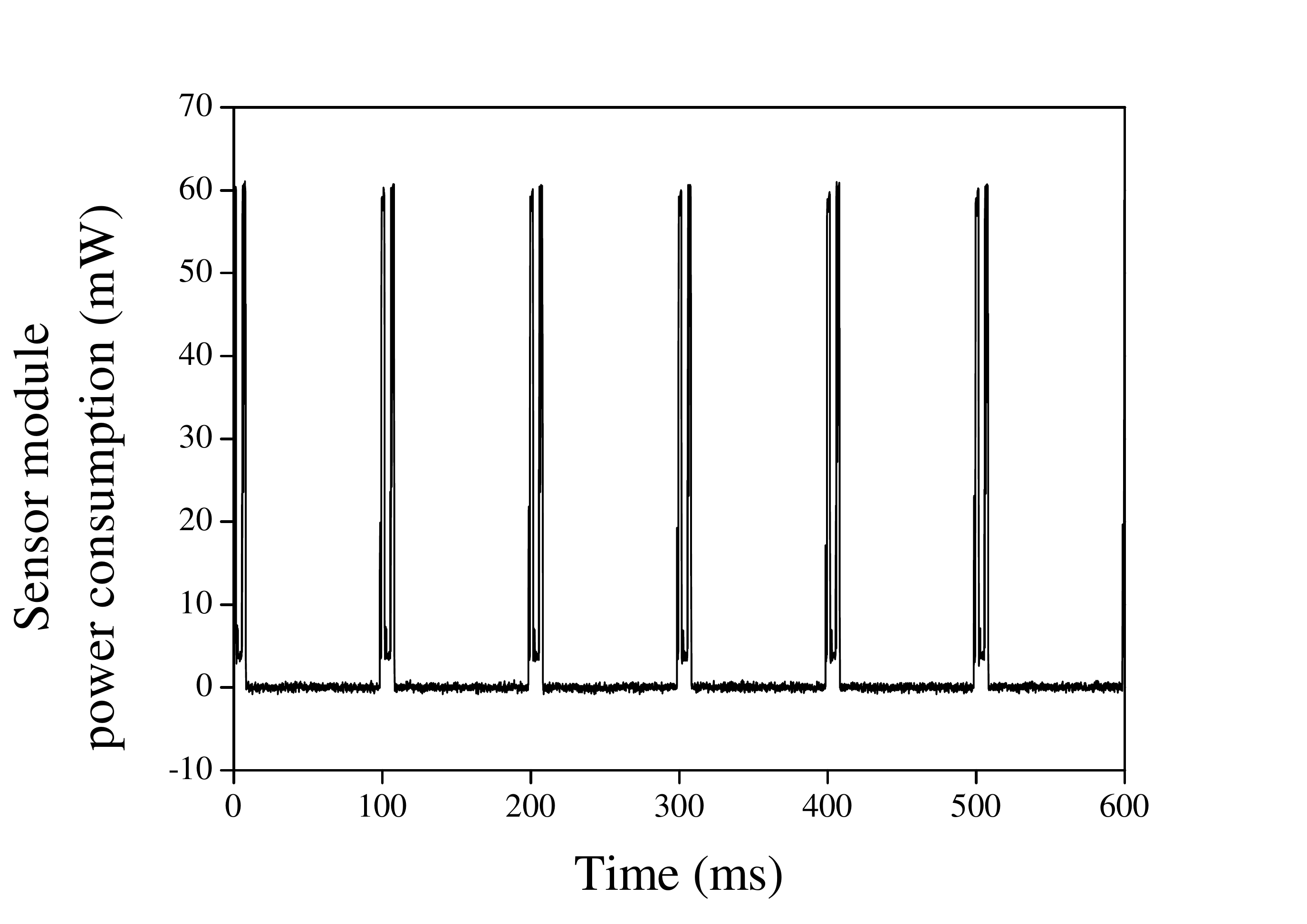}
        }~
    \subfigure[Short time period]{
        \label{fig:sensconsshort}\includegraphics[width=4.1cm, bb=0.7in 0.3in 9.5in 7.4in] {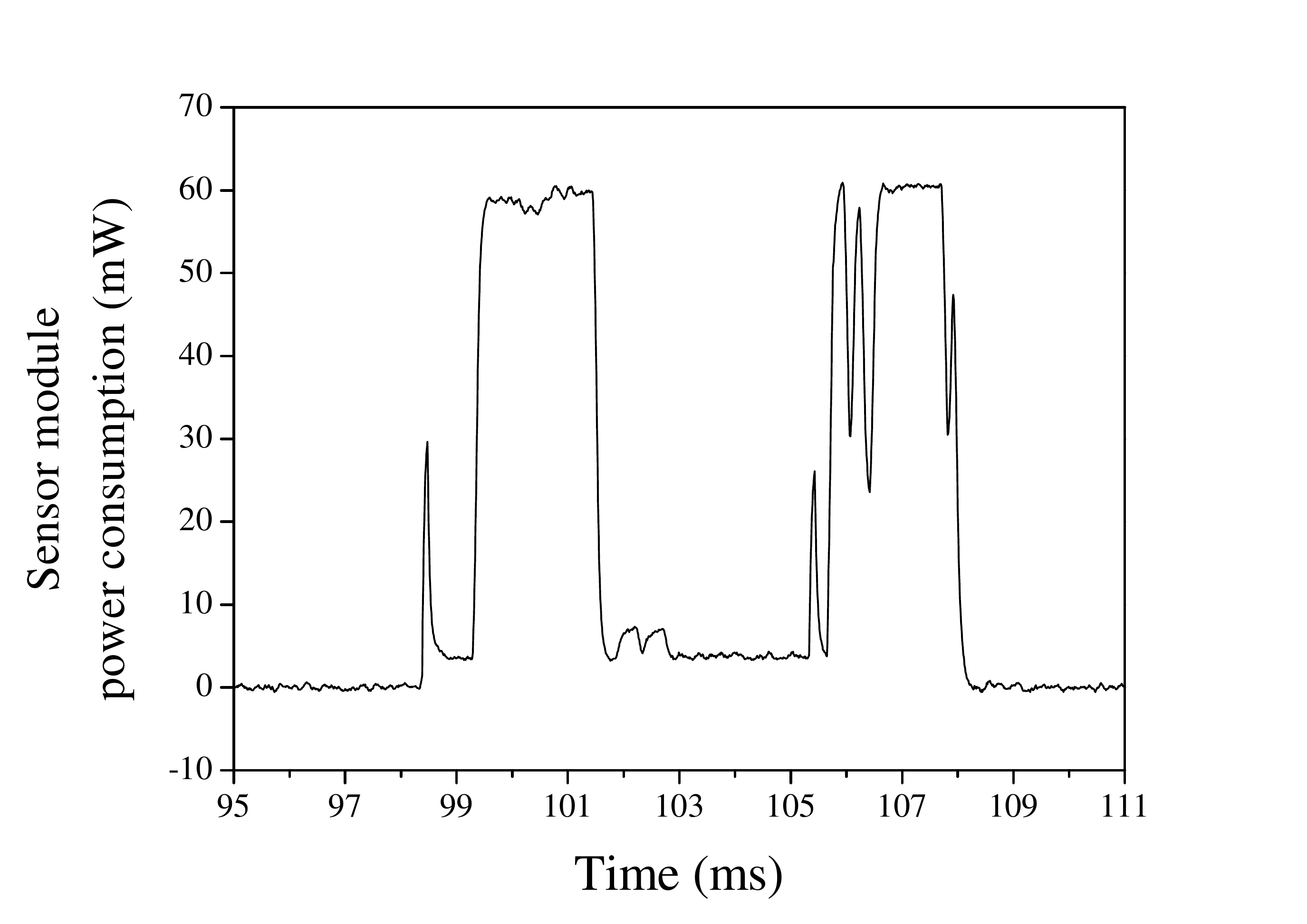}
        }
    \caption{Sensor module power consumption over time.}
    \label{fig:senscons}
\end{figure}

\begin{figure}
	\centering
    \includegraphics[width=4.5cm, bb=0.8in 0.3in 9.6in 7.2in]{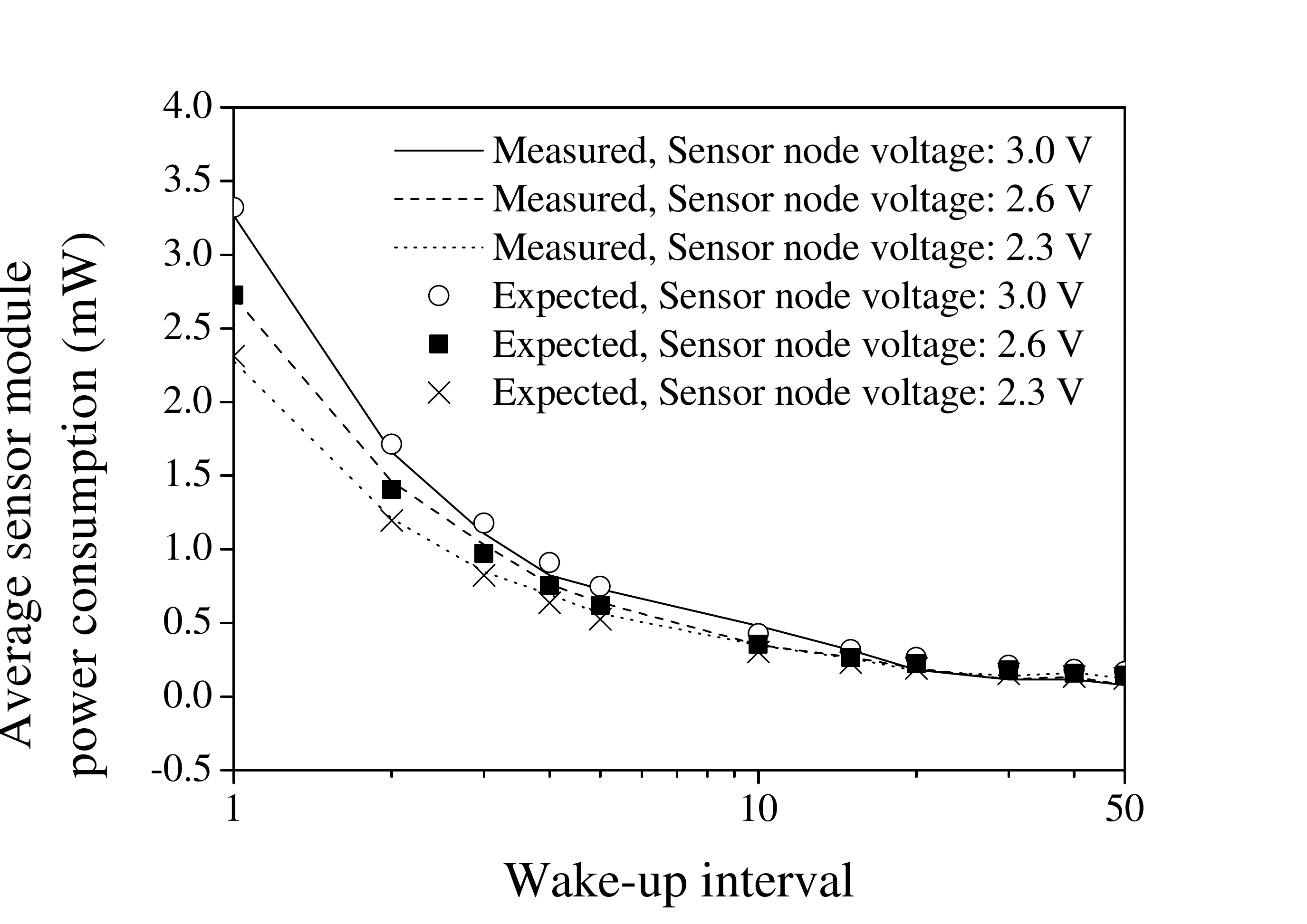}
    \caption{Average sensor module power consumption according to wake-up interval.}
    \label{fig:intpwr}
\end{figure}

Fig.~\ref{fig:senschar} shows the energy storage charging power over time.
The energy storage charging power is defined as the power charged to the energy storage module, which is the multiplication of the sensor node voltage and the current through the energy storage module, i.e., $V I_\text{ES}$.
For this figure, we set $\alpha = 0.4$, $\Upsilon = 2.3$ W, $\tau = 1$, the power attenuation to 15 dB, and the sensor node voltage to 3.0 V.
In Fig.~\ref{fig:senschar}, we can see that the power is discharged around the start of each frame due to the sensor module power consumption and the power is charged after the sensor module is put into an idle mode.
In Fig.~\ref{fig:harvcons}, we show the energy storage charging power according to the energy transfer power and the wake-up interval when $\alpha = 1$, the sensor node voltage is 3.0 V, and the power attenuations are 26.89 dB and 30.62 dB.
Considering that the supercapacitor leakage power is small, the energy storage can be stable or charged as long as the energy storage charging power is non-negative.
In Fig.~\ref{fig:harvcons}, we can see that high energy transfer power or high wake-up interval is required for positive energy storage charging power.

\begin{figure}
    \centering
    \subfigure[Long time period]{
        \label{fig:senscharlong}\includegraphics[width=4.1cm, bb=0.7in 0.3in 9.5in 7.4in] {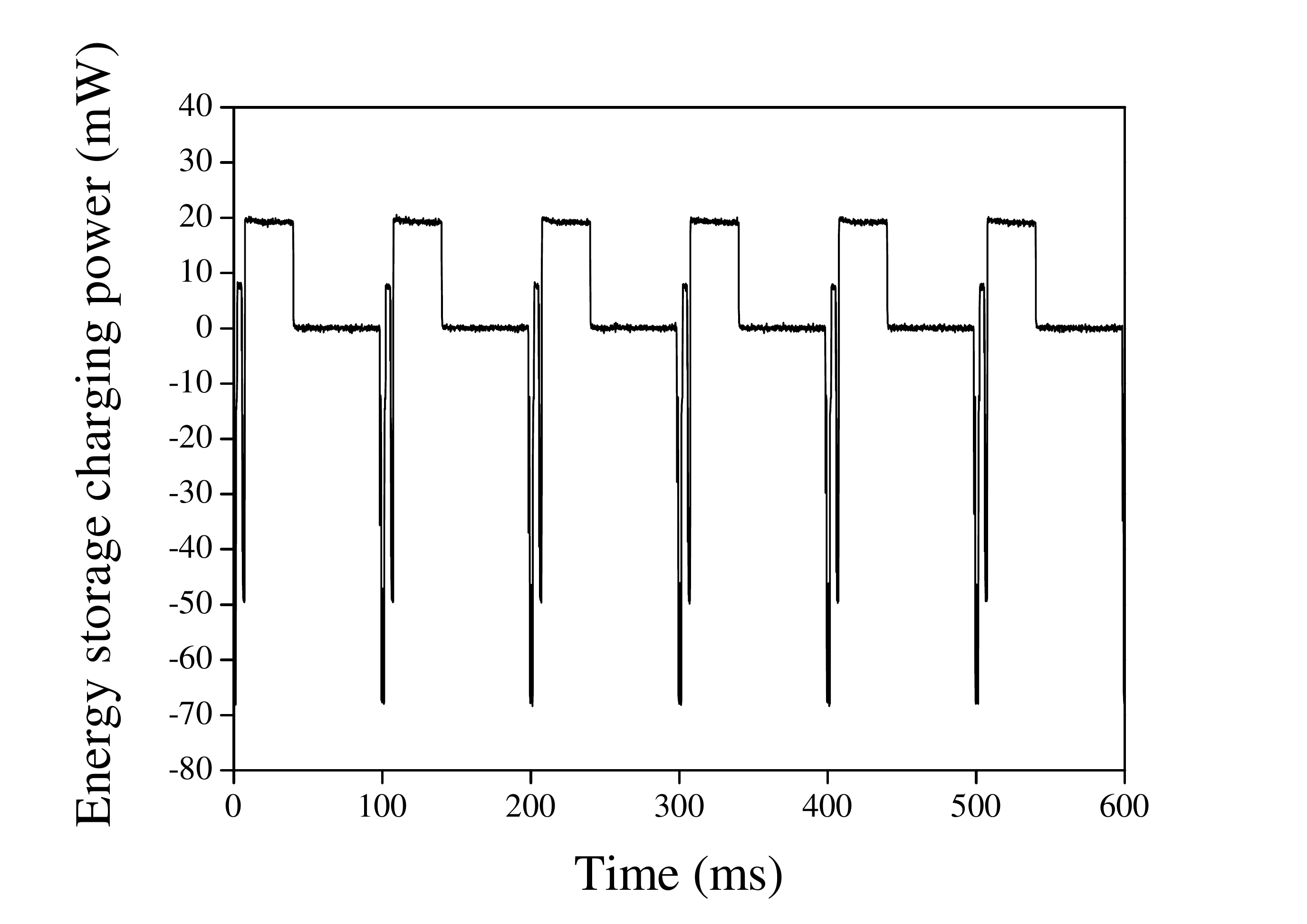}
        }~
    \subfigure[Short time period]{
        \label{fig:senscharshort}\includegraphics[width=4.1cm, bb=0.7in 0.3in 9.5in 7.4in] {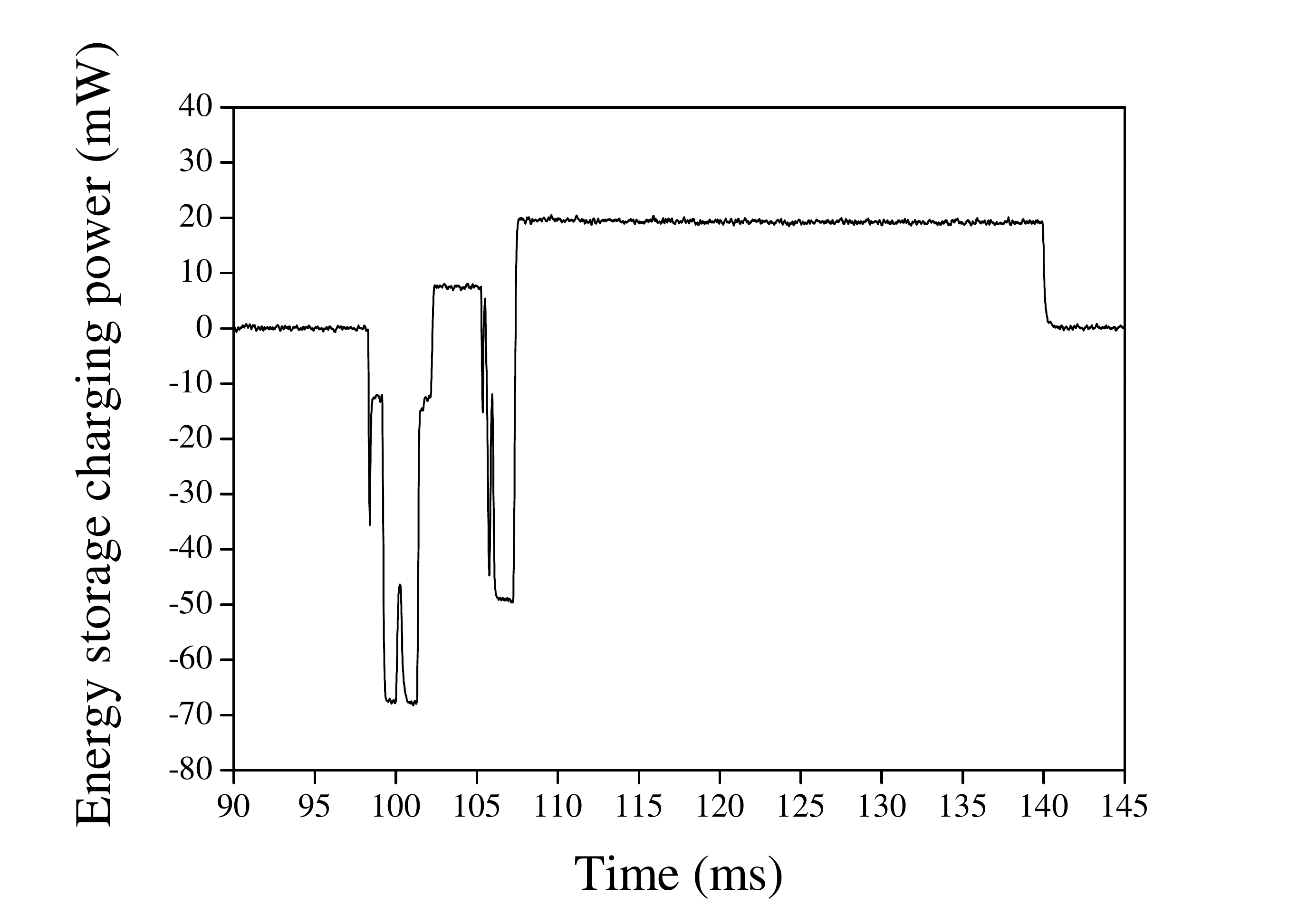}
        }
    \caption{Energy storage charging power over time.}
    \label{fig:senschar}
\end{figure}

\begin{figure}
    \centering
    \subfigure[Attenuation: 26.89 dB]{
        \label{fig:harvcons1}\includegraphics[width=4.1cm, bb=0.7in 0.3in 10in 7.4in] {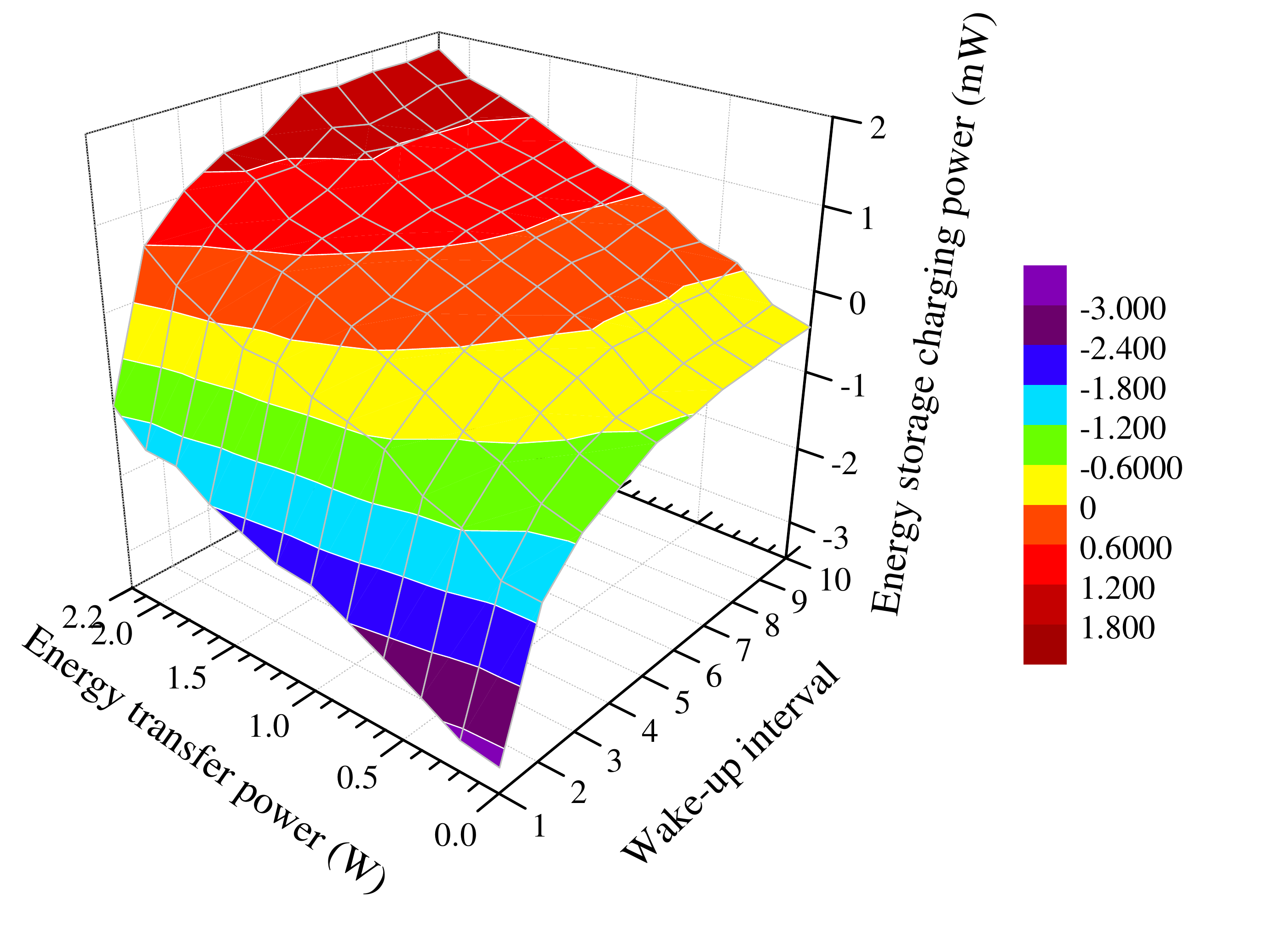}
        }~
    \subfigure[Attenuation: 30.62 dB]{
        \label{fig:harvcons2}\includegraphics[width=4.1cm, bb=0.7in 0.3in 10in 7.4in] {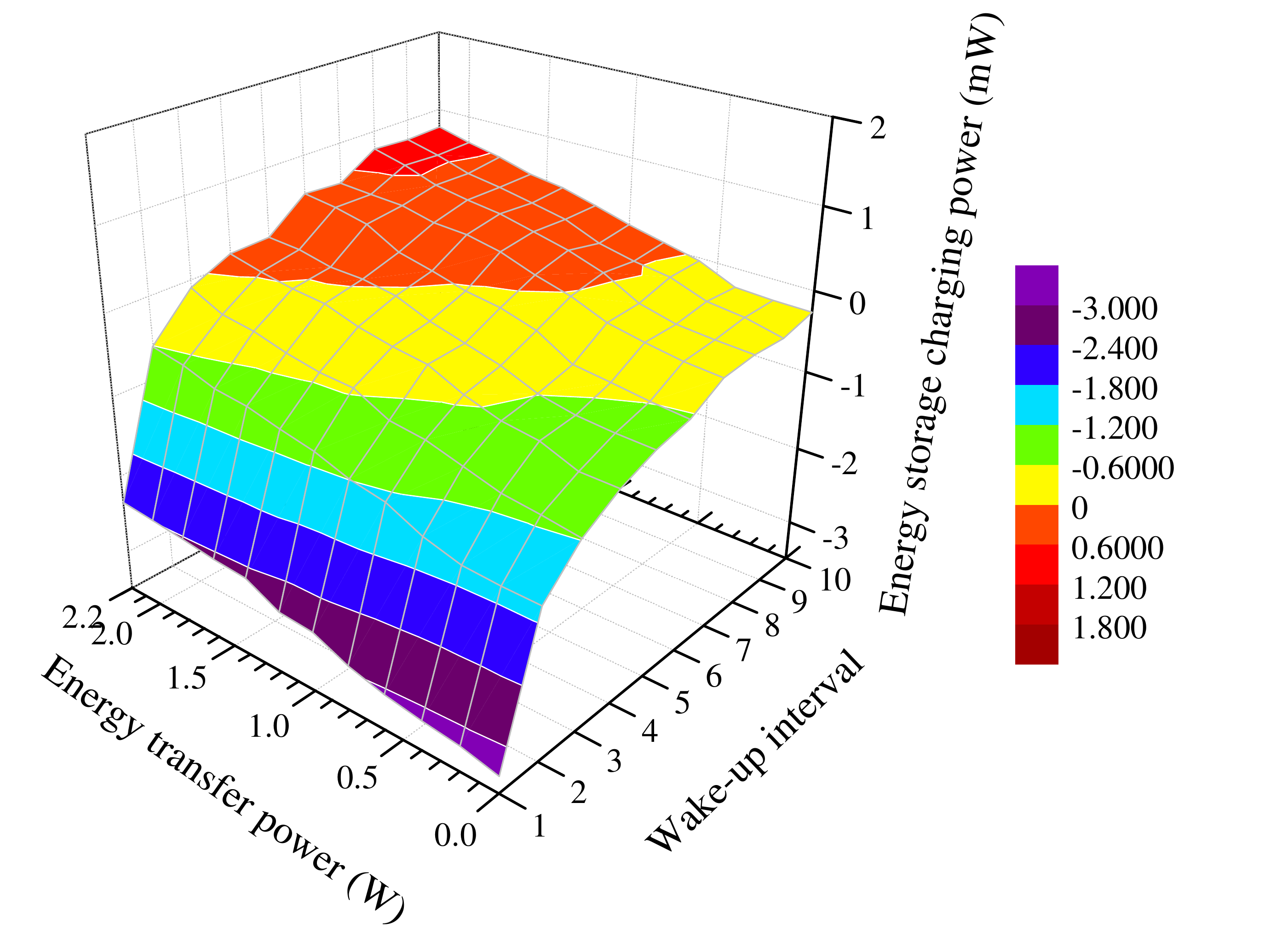}
        }
    \caption{Energy storage charging power according to energy transfer power and wake-up interval.}
    \label{fig:harvcons}
\end{figure}

In Fig.~\ref{fig:harvconsexmodelval}, we validate the discrete-time stored energy evolution model in Section \ref{section:evolutionmodel} by comparing the expected and the measured results.
By slightly modifying \eqref{eq:delta}, the expected energy storage charging power is calculated as $\Phi(E,\alpha,\Upsilon,h)/T_\text{frame} - (\sum_{m\in \{\text{rx},\text{act},\text{tx},\text{idle}\}} \xi_m(E)\cdot T_m
 + \xi_\text{idle}(E) \cdot (\tau-1) T_\text{frame})/(\tau T_\text{frame})$.
In Fig.~\ref{fig:harvconsex}, we can see that the expected energy storage charging power is very similar to the measured one.
Fig.~\ref{fig:modelval} compares the expected and the measured stored energy variation over time.
The expected stored energy variation is calculated as $(\Phi(E,\alpha,\Upsilon,h)-Q(E,1/\tau))/T_\text{frame}$.
In calculating the expected stored energy variation, the stored energy $E$ is the only measured parameter and all other parameters and functions are given by the stored energy evolution model.
To calculate the measured stored energy variation, we have measured the stored energy every second and have taken the difference of two consecutively obtained stored energy measurements.
For Fig.~\ref{fig:modelval}, the power attenuation is set to 26.89 dB and we change the parameters $(\alpha,\Upsilon\text{ (in Watt)},\tau)$ every ten seconds in the following sequence: (0,2,1), (1,2,1), (1,2,2), (1,1,2), (0.5,1,2), (0.2,1,2), (0.2,1,5), (0.2,1,10), (1,1,10).
In this figure, we can see that the expected stored energy variation well approximates the measured one, which proves the validity of our stored energy evolution model.

\begin{figure}
    \centering
    \subfigure[Energy storage charging power]{
        \label{fig:harvconsex}\includegraphics[width=4.1cm, bb=0.7in 0.3in 9.5in 7.4in] {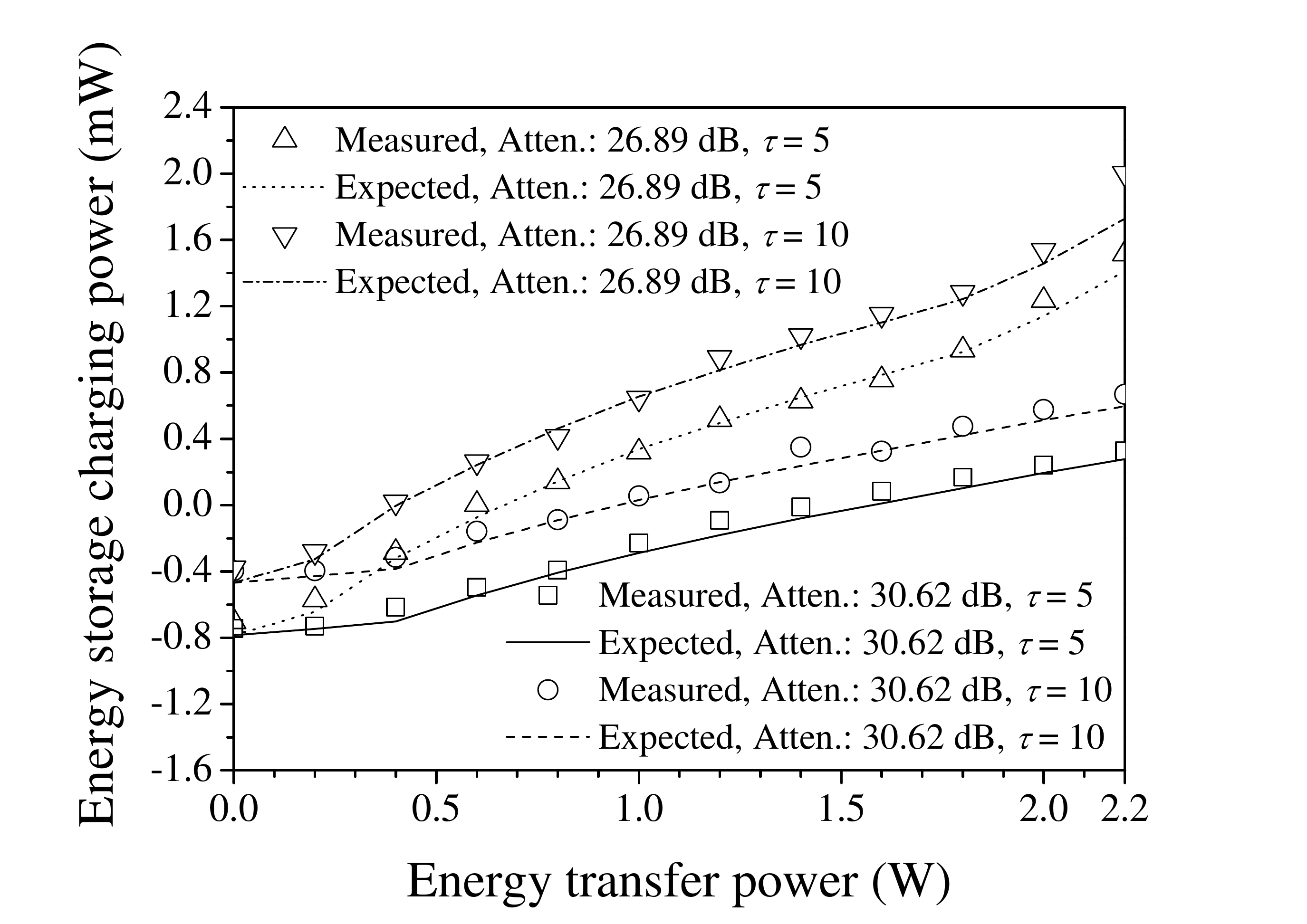}
        }~
    \subfigure[Stored energy variation]{
        \label{fig:modelval}\includegraphics[width=4.1cm, bb=0.7in 0.3in 9.5in 7.4in] {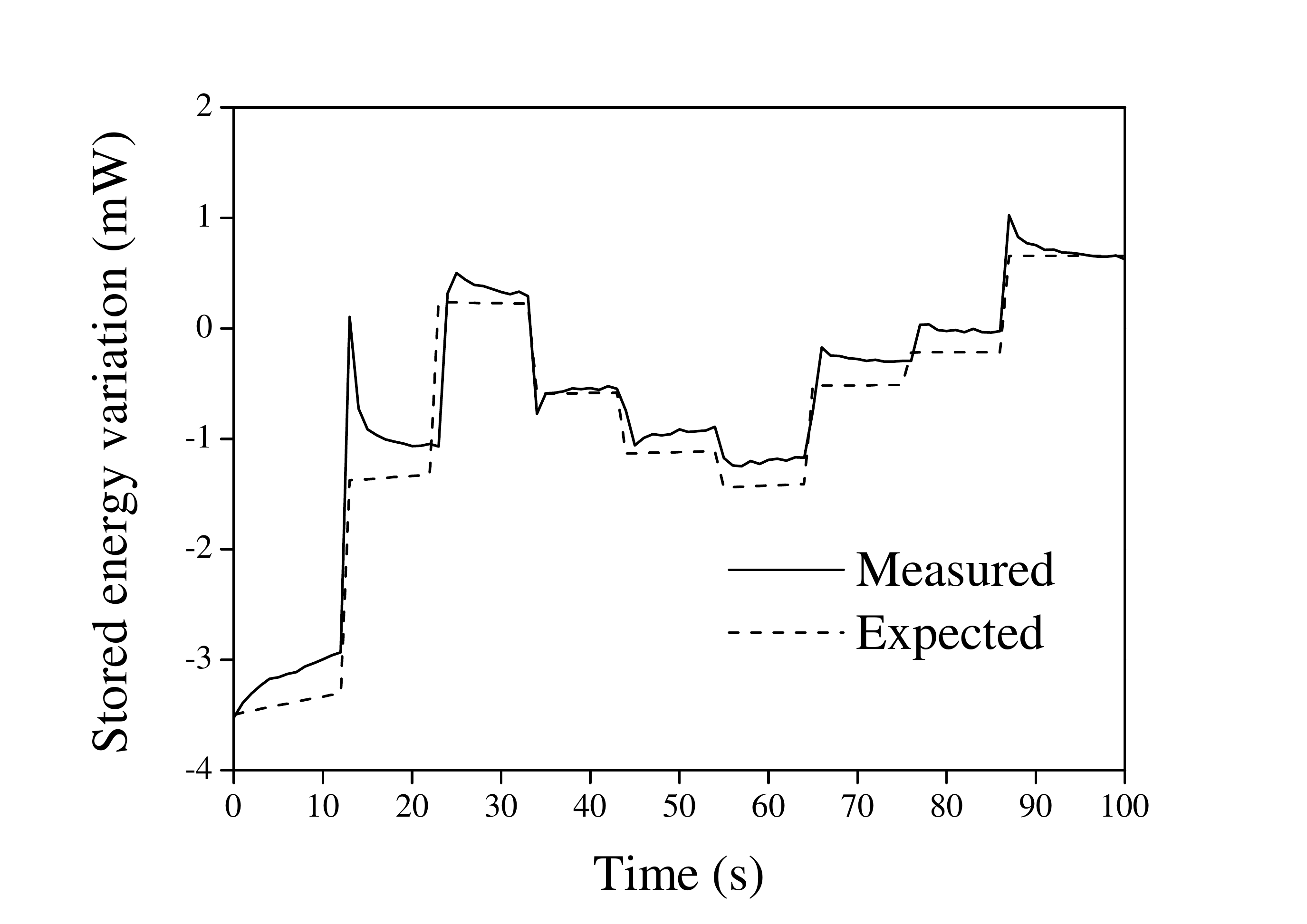}
        }
    \caption{Measured and expected energy storage charging power and stored energy variation.}
    \label{fig:harvconsexmodelval}
\end{figure}

In Figs.~\ref{fig:conttime} and \ref{fig:conttimeant}, we show the operation of the adaptive energy management algorithm in Section \ref{section:emalg}.
In these figures, we plot the amplifier duty cycle, the energy transfer power, the wake-up interval, the stored energy, and the amplifier power consumption over time when the adaptive energy management algorithm is used.
The algorithm parameters are $\alpha_\text{min}=0.1$, $\Upsilon_\text{max}=2300$ mW, $\tau_\text{tgt}=1$, $E_\text{tgt}=380$ mJ, $C_P=0$, $C_I=0.01$, $\beta_{\Upsilon}=100$, and $\beta_{\tau}=1$.
We use the step attenuator-based channel environment for Fig.~\ref{fig:conttime} and the antenna-based channel environment for Fig.~\ref{fig:conttimeant}.
We set $S_\text{tgt}=20$ mW for Fig.~\ref{fig:conttime} and $S_\text{tgt}=10$ mW for Fig.~\ref{fig:conttimeant}.
For Fig.~\ref{fig:conttime}, we set the power attenuation (in dB) every 100 seconds in the following sequence: 16.69, 20.66, 24.63, 28.62, 32.82, 16.69, 32.82, 24.63, 20.66, 16.69.
For Fig.~\ref{fig:conttimeant}, the distance (in meter) is changed every 100 seconds in the following sequence: 1.0, 1.5, 2.0, 2.5, 3.0, 3.5, 2.5, 1.0, 3.0, 2.0.
In both figures, we can see that the energy management algorithm controls one of $\alpha$, $\Upsilon$, and $\tau$ at a time in order to keep the stored energy $E_i$ to $E_\text{tgt}$.
As a result, we can see that the stored energy is maintained to be around 380 mJ.

\begin{figure}
    \centering
    \subfigure[Amplifier duty cycle, energy transfer power, and wake-up interval]{
        \label{fig:conttime1}\includegraphics[width=4.1cm, bb=1.7in 0.3in 10.1in 7.4in] {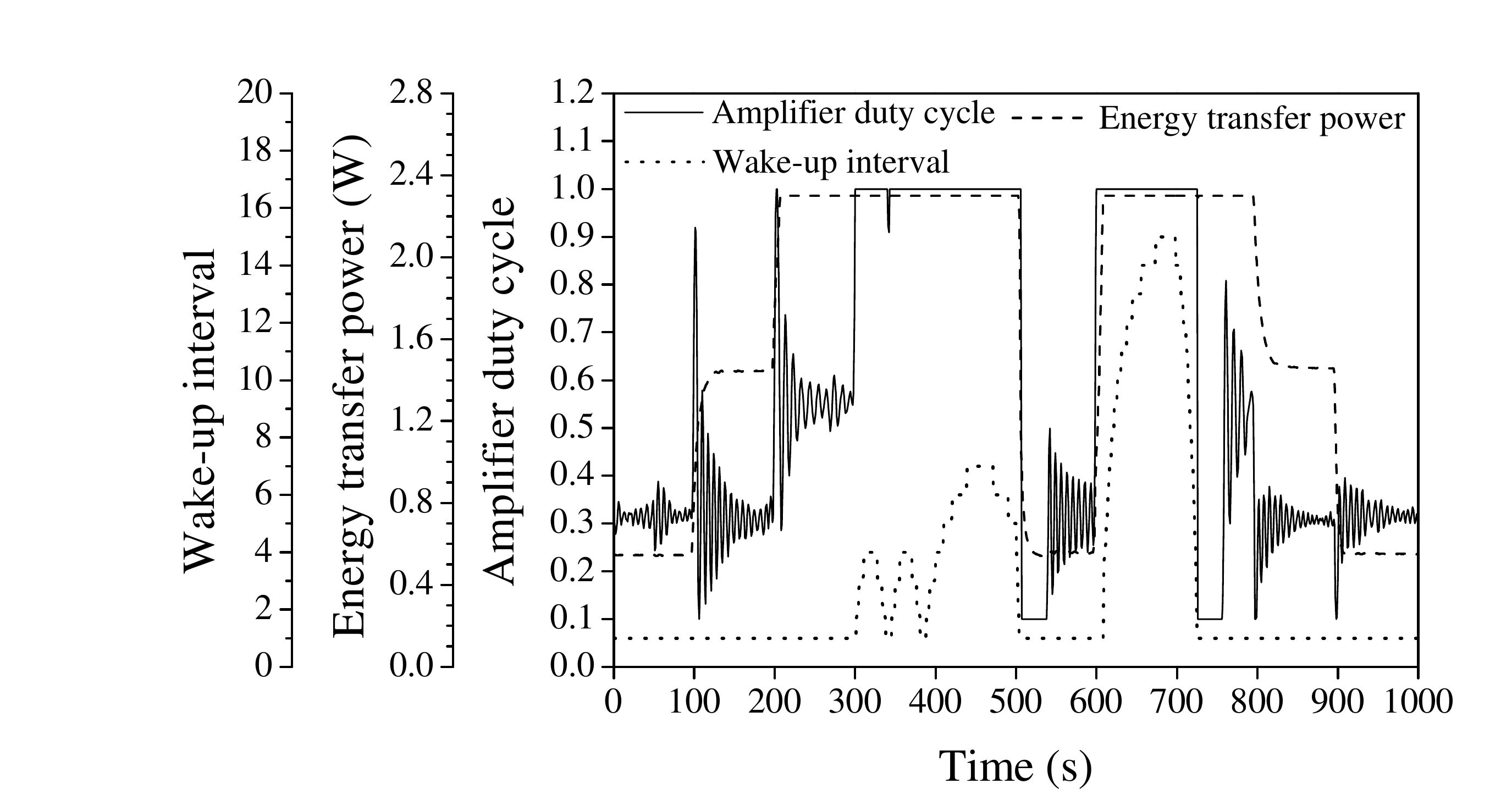}
        }~
    \subfigure[Stored energy and amplifier power consumption]{
        \label{fig:conttime2}\includegraphics[width=4.1cm, bb=2.1in 0.3in 10.5in 7.4in] {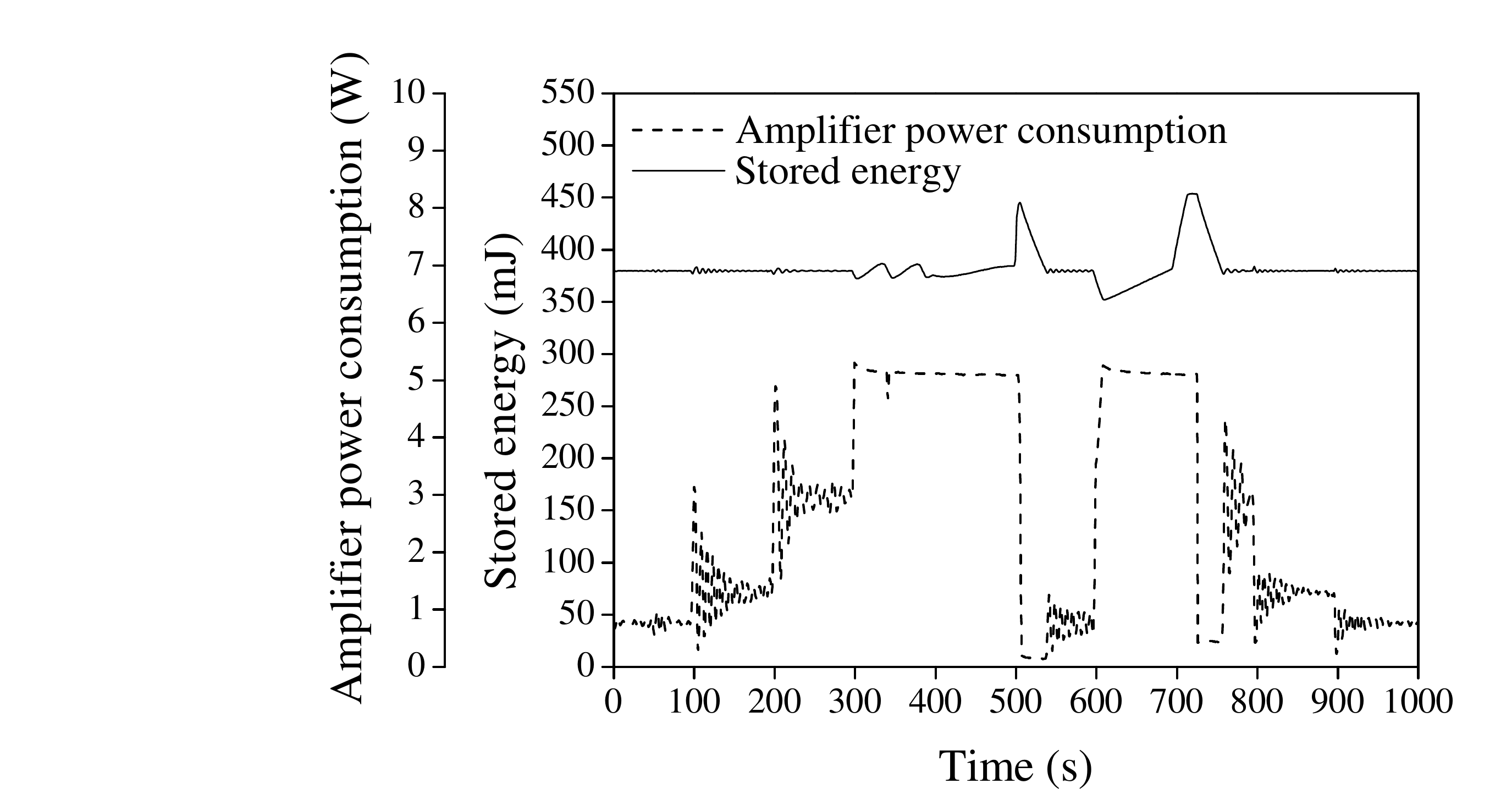}
        }
    \caption{Operation of proposed energy management algorithm over time in attenuator-based channel environment.}
    \label{fig:conttime}
\end{figure}

\begin{figure}
    \centering
    \subfigure[Amplifier duty cycle, energy transfer power, and wake-up interval]{
        \label{fig:conttime1ant}\includegraphics[width=4.1cm, bb=1.7in 0.3in 10.1in 7.4in] {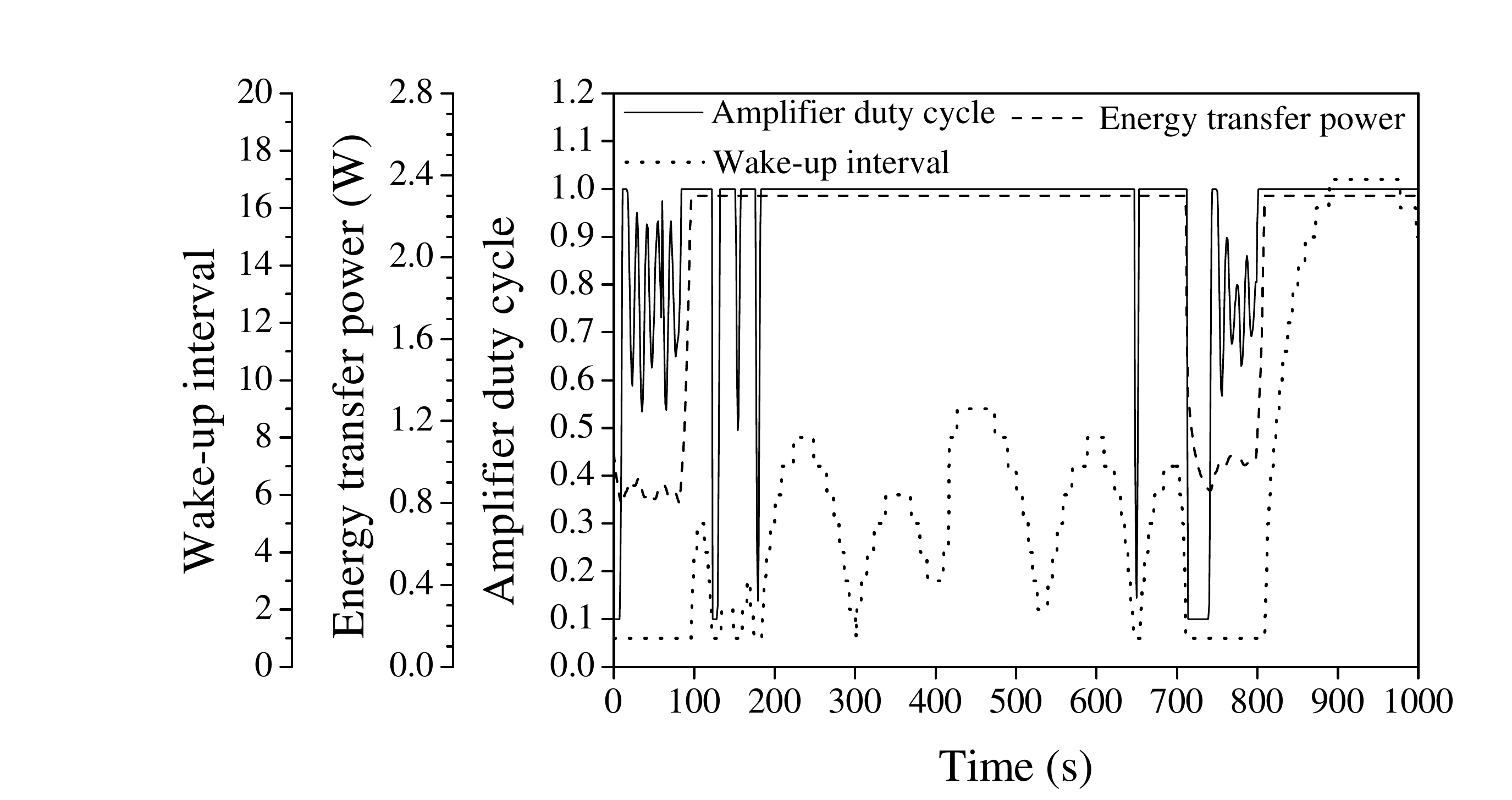}
        }~
    \subfigure[Stored energy and amplifier power consumption]{
        \label{fig:conttime2ant}\includegraphics[width=4.1cm, bb=2.1in 0.3in 10.5in 7.4in] {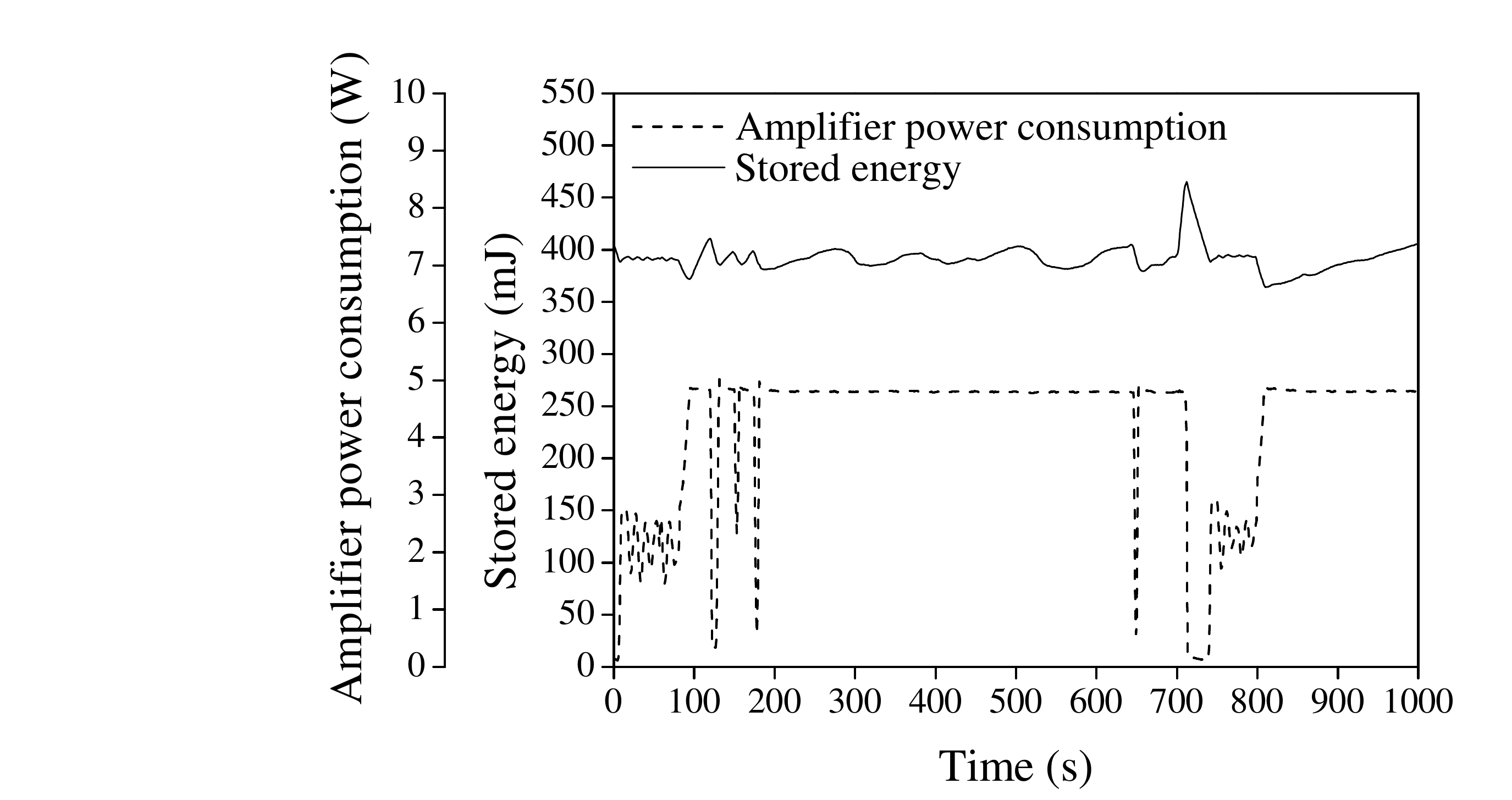}
        }
    \caption{Operation of proposed energy management algorithm over time in antenna-based channel environment.}
    \label{fig:conttimeant}
\end{figure}

Figs.~\ref{fig:energymanagementresult} and \ref{fig:contamppow} compare the performance of the proposed energy management scheme and the no amplifier duty cycling scheme.
The no amplifier duty cycling scheme fixes the amplifier duty cycle to one (i.e., $\alpha = 1$) while it controls $\Upsilon$ and $\tau$ in the same way as the proposed algorithm does.
Fig.~\ref{fig:energymanagementresult} plots $\alpha$, $\Upsilon$, and $\tau$ after convergence according to the power attenuation.
For Fig.~\ref{fig:contprop20}, we set $S_\text{tgt}=20$ mW.
In Fig.~\ref{fig:contprop20}, we can see that $\Upsilon$ is first controlled and then $\alpha$ and $\tau$ are controlled as the power attenuation becomes severe.
Fig.~\ref{fig:contamppow} shows the amplifier power consumption of the proposed energy management scheme and the no amplifier duty cycling scheme.
In this figure, we can see that the proposed scheme outperforms the no amplifier duty cycling scheme.
When $S_\text{tgt}=20$ mW and the power attenuation is less than 23 dB, the proposed scheme only consumes less than half of the power consumed by the no amplifier duty cycling scheme.

\begin{figure}
    \centering
    \subfigure[Proposed energy management algorithm]{
        \label{fig:contprop20}\includegraphics[width=4.1cm, bb=1.7in 0.3in 10.1in 7.4in] {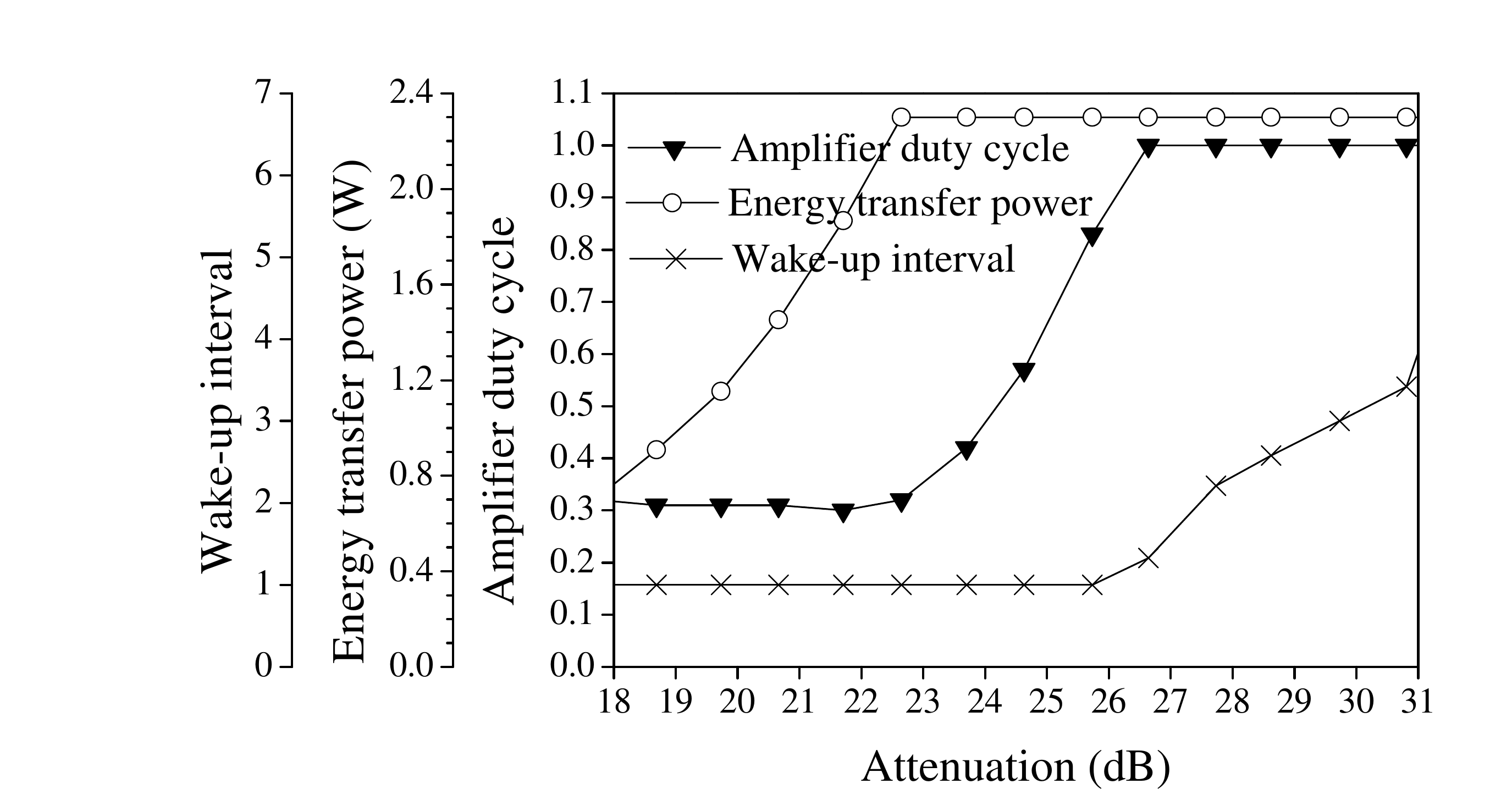}
        }~
    \subfigure[no amplifier duty cycling scheme]{
        \label{fig:contcomp}\includegraphics[width=4.1cm, bb=1.7in 0.3in 10.1in 7.4in] {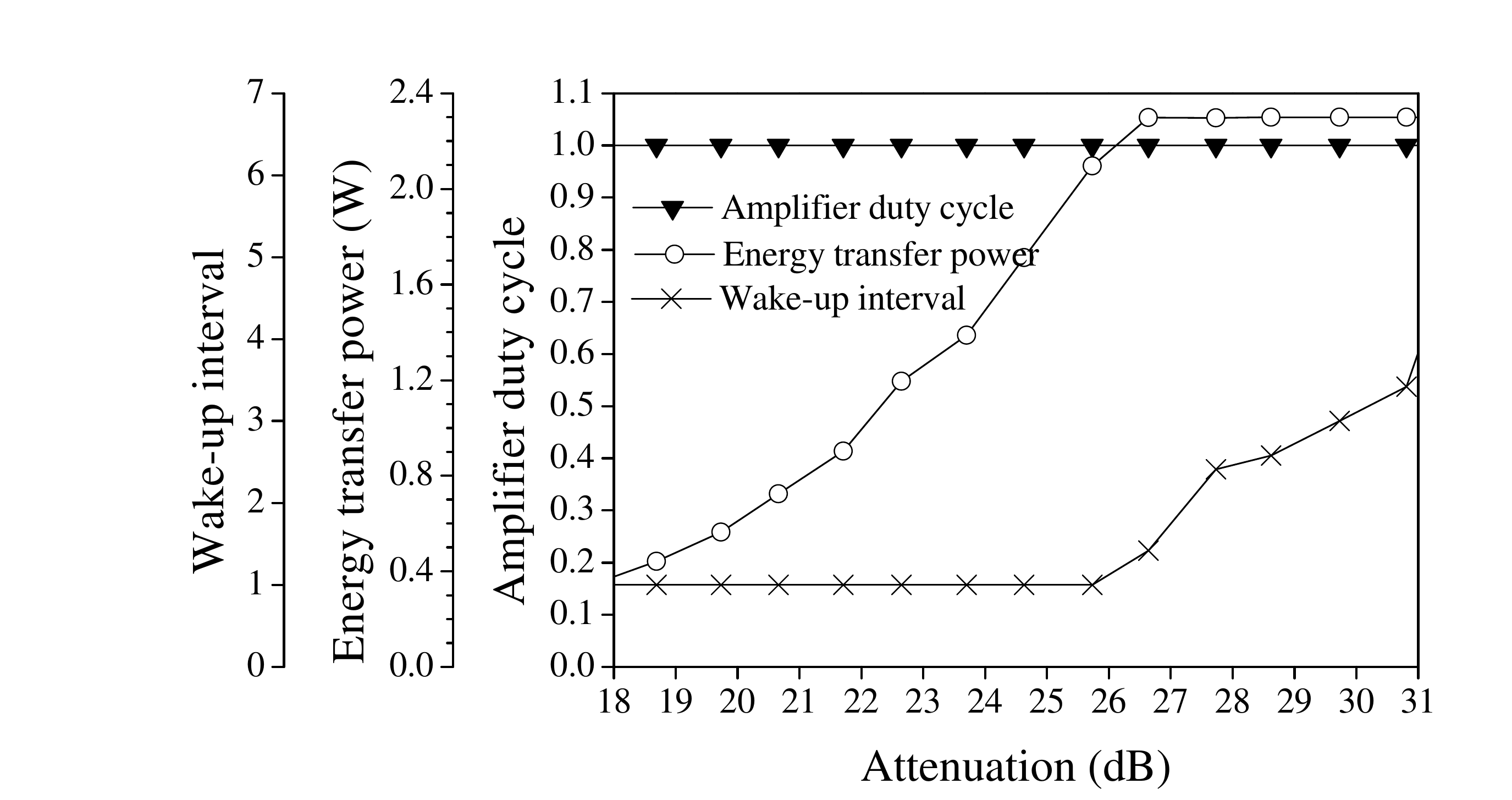}
        }
    \caption{Amplifier duty cycle, energy transfer power, and wake-up interval of proposed energy management scheme and no amplifier duty cycling scheme.}
    \label{fig:energymanagementresult}
\end{figure}

\begin{figure}
	\centering
    \includegraphics[width=4.5cm, bb=0.8in 0.3in 9.6in 7.2in]{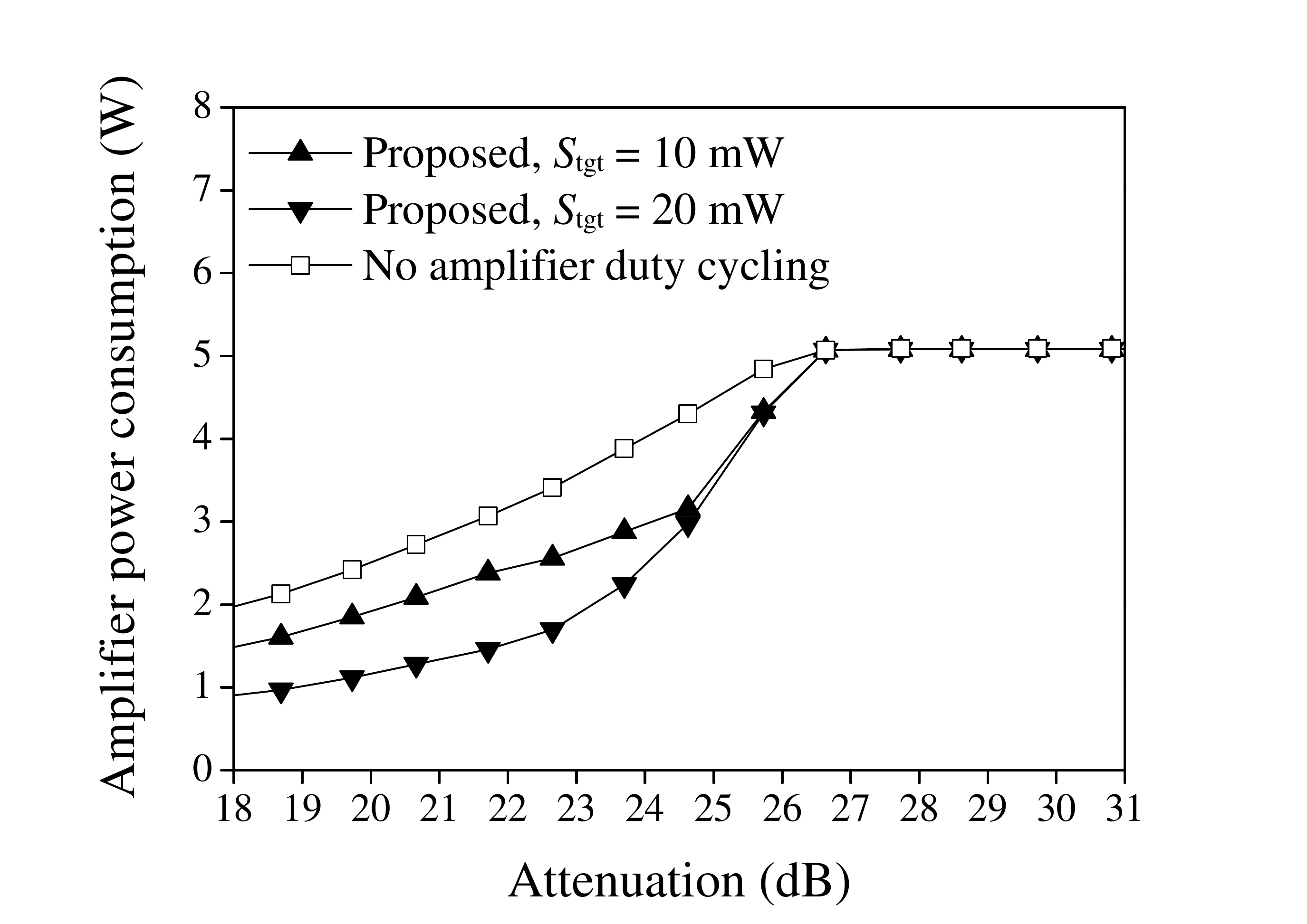}
    \caption{Comparison of amplifier power consumption of proposed energy management scheme and no amplifier duty cycling scheme.}
    \label{fig:contamppow}
\end{figure}

\section{Conclusion}\label{section:conclusion}

In this paper, we have built a full-fledged WPSN testbed and have conducted extensive experiments on the testbed.
By using the testbed, we have obtained the parameters for the WPSN model and have validated the usefulness and practical importance of the model.
Based on the stored energy evolution model, we have proposed the energy management scheme for the energy neutral operation of the sensor node.
The experimental results demonstrated that the proposed scheme adaptively controls the WPSN to achieve energy neutrality and high-efficiency RF power transfer.

\appendices

\section{Proof of Theorem \ref{theorem:optimal}}\label{proof:optimal}

If $Q(E_\text{tgt},r_\text{tgt}) \le \widehat{H}$, we can define $\widehat{\alpha} =  Q(E_\text{tgt},r_\text{tgt})/\widehat{H}$ such that $0\le \widehat{\alpha} \le 1$.
For such $\widehat{\alpha}$, it is satisfied that $\Phi(E_\text{tgt},\widehat{\alpha},\widehat{\Upsilon},h) = \widehat{\alpha}\widehat{H} = Q(E_\text{tgt},r_\text{tgt})$.
According to the definition of $\widehat{\Upsilon}$, we have $L(\widehat{\alpha},\widehat{\Upsilon},\widehat{\mu})=g(\widehat{\mu})\le \Omega(\alpha^*,\Upsilon^*)$.
In addition, we also have $L(\widehat{\alpha},\widehat{\Upsilon},\widehat{\mu})=\Omega(\widehat{\alpha},\widehat{\Upsilon}) - \widehat{\mu}(\Phi(E_\text{tgt},\widehat{\alpha},\widehat{\Upsilon},h)- Q(E_\text{tgt},r_\text{tgt}))=\Omega(\widehat{\alpha},\widehat{\Upsilon})$.
Therefore, we can conclude that $\Omega(\widehat{\alpha},\widehat{\Upsilon}) \le \Omega(\alpha^*,\Upsilon^*)$, which leads to $\widehat{\alpha}=\alpha^*$ and $\widehat{\Upsilon} = \Upsilon^*$.

\bibliographystyle{IEEEtran}
\bibliography{IEEEabrv,WPSN_BIB}

\begin{thebibliography}{10}
\providecommand{\url}[1]{#1}
\csname url@samestyle\endcsname
\providecommand{\newblock}{\relax}
\providecommand{\bibinfo}[2]{#2}
\providecommand{\BIBentrySTDinterwordspacing}{\spaceskip=0pt\relax}
\providecommand{\BIBentryALTinterwordstretchfactor}{4}
\providecommand{\BIBentryALTinterwordspacing}{\spaceskip=\fontdimen2\font plus
\BIBentryALTinterwordstretchfactor\fontdimen3\font minus
  \fontdimen4\font\relax}
\providecommand{\BIBforeignlanguage}[2]{{%
\expandafter\ifx\csname l@#1\endcsname\relax
\typeout{** WARNING: IEEEtran.bst: No hyphenation pattern has been}%
\typeout{** loaded for the language `#1'. Using the pattern for}%
\typeout{** the default language instead.}%
\else
\language=\csname l@#1\endcsname
\fi
#2}}
\providecommand{\BIBdecl}{\relax}
\BIBdecl

\bibitem{Huang:2015}
K.~Huang and X.~Zhou, ``Cutting the last wires for mobile communications by
  microwave power transfer,'' \emph{{IEEE} Commun. Mag.}, vol.~53, no.~6, pp.
  86--93, 2015.

\bibitem{Xie:2013}
L.~Xie, Y.~Shi, Y.~T. Hou, and W.~Lou, ``Wireless power transfer and
  applications to sensor networks,'' \emph{{IEEE} Wireless Commun.}, vol.~20,
  no.~4, pp. 140–--145, Aug. 2013.

\bibitem{Bi:2015}
S.~Bi, C.~K. Ho, and R.~Zhang, ``Wireless powered communication: opportunities
  and challenges,'' \emph{{IEEE} Commun. Mag.}, vol.~53, no.~4, pp. 117--125,
  Apr. 2015.

\bibitem{Lu:2016}
X.~Lu, P.~Wang, D.~Niyato, D.~I. Kim, and Z.~Han, ``Wireless charging
  technologies: Fundamentals, standards, and network applications,''
  \emph{{IEEE} Commun. Surveys Tuts.}, vol.~18, no.~2, pp. 1413--1452, Second
  Quarter 2016.

\bibitem{Lu:2015Dec}
X.~Lu, P.~Wang, D.~Niyato, and Z.~Han, ``Resource allocation in wireless
  networks with {RF} energy harvesting and transfer,'' \emph{{IEEE} Netw.},
  vol.~29, no.~6, pp. 68--75, Dec. 2015.

\bibitem{Choi:2015}
K.~W. Choi and D.~I. Kim, ``Stochastic optimal control for wireless powered
  communication networks,'' \emph{{IEEE} Trans. Wireless Commun.}, vol.~15,
  no.~1, pp. 686--698, Jan. 2016.

\bibitem{Chen:2015}
H.~Chen, Y.~Li, J.~L. Rebelatto, B.~Uchoa-Filho, and B.~Vucetic,
  ``Harvest-then-cooperate: Wireless-powered cooperative communications,''
  \emph{{IEEE} Trans. Signal Process.}, vol.~63, no.~7, 2015.

\bibitem{Kang:2015}
X.~Kang, C.~K. Ho, and S.~Sun, ``Full-duplex wireless-powered communication
  network with energy causality,'' \emph{{IEEE} Trans. Wireless Commun.},
  vol.~14, no.~10, pp. 5539--5551, Oct. 2015.

\bibitem{Dolgov:2010}
A.~Dolgov, R.~Zane, and Z.~Popovic, ``Power management system for online low
  power {RF} energy harvesting optimization,'' \emph{{IEEE} Trans. Circuits
  Syst. {I}}, vol.~57, no.~7, pp. 1802--1811, Jul. 2010.

\bibitem{Popovic:2013}
Z.~Popovic, E.~A. Falkenstein, D.~Costinett, and R.~Zane, ``Low-power far-field
  wireless powering for wireless sensors,'' \emph{Proc. {IEEE}}, vol. 101,
  no.~6, pp. 1397–--1409, Jun. 2013.

\bibitem{Visser:2012}
H.~J. Visser, ``Indoor wireless {RF} energy transfer for powering wireless
  sensors,'' \emph{Radioengineering}, vol.~21, no.~4, pp. 963--973, Dec. 2012.

\bibitem{Farinholt:2009}
K.~M. Farinholt, G.~Park, and C.~R. Farrar, ``{RF} energy transmission for a
  low-power wireless impedance sensor node,'' \emph{{IEEE} Sensors J.}, vol.~9,
  no.~7, pp. 793--800, Jul. 2009.

\bibitem{Sudevalayam:2011}
S.~Sudevalayam and P.~Kulkarni, ``Energy harvesting sensor nodes: Survey and
  implications,'' \emph{{IEEE} Commun. Surveys Tuts.}, vol.~13, no.~3, pp.
  443--461, Third Quarter 2011.

\bibitem{Moser:2010}
C.~Moser, L.~Thiele, D.~Brunelli, and L.~Benini, ``Adaptive power management
  for environmentally powered systems,'' \emph{{IEEE} Trans. Comput.}, vol.~59,
  no.~4, pp. 478–--491, Apr. 2010.

\bibitem{Renner:2014}
C.~Renner, S.~Unterschütz, V.~Turau, K.~Romer, and S.~U. Utz, ``Perpetual data
  collection with energy-harvesting sensor networks,'' \emph{ACM Trans. Sens.
  Networks}, vol.~11, no.~1, pp. 12:1--12:45, Dec. 2014.

\bibitem{Shigeta:2013}
R.~Shigeta, T.~Sasaki, D.~M. Quan, Y.~Kawahara, R.~J. Vyas, M.~M. Tentzeris,
  and T.~Asami, ``Ambient {RF} energy harvesting sensor device with
  capacitor-leakage-aware duty cycle control,'' \emph{{IEEE} Sensors J.},
  vol.~13, no.~8, pp. 2973–--2983, Aug. 2013.

\bibitem{Vyas:2013}
R.~J. Vyas, B.~B. Cook, Y.~Kawahara, and M.~M. Tentzeris, ``{E-WEHP}: A
  batteryless embedded sensor-platform wirelessly powered from ambient
  digital-{TV} signals,'' \emph{{IEEE} Trans. Microw. Theory Techn.}, vol.~61,
  no.~6, pp. 2491--2505, Jun. 2013.

\bibitem{Choi:2016}
K.~W. Choi, L.~Ginting, P.~A. Rosyady, A.~A. Aziz, and D.~I. Kim,
  ``Wireless-powered sensor networks: How to realize,'' \emph{{IEEE} Trans.
  Wireless Commun.}, 2016, to be published.

\bibitem{Mishra:2015}
D.~Mishra, S.~De, and K.~R. Chowdhury, ``Charging time characterization for
  wireless {RF} energy transfer,'' \emph{{IEEE} Trans. Circuits Syst. {II}},
  vol.~62, no.~4, pp. 362--366, Apr. 2015.

\bibitem{zolertiaz1}
\BIBentryALTinterwordspacing
{Zolertia}, ``{Z1 datasheet}.'' [Online]. Available:
  \url{http://zolertia.sourceforge.
  net/wiki/images/e/e8/Z1\_RevC\_Datasheet.pdf}
\BIBentrySTDinterwordspacing

\bibitem{p1110}
\BIBentryALTinterwordspacing
{Powercast Corporation}, ``{Product datasheet: P1110--915 MHz RF Powerharvester
  Receiver}.'' [Online]. Available: \url{http://www.powercastco.com/
  PDF/P1110-datasheet.pdf}
\BIBentrySTDinterwordspacing

\bibitem{cc2420}
\BIBentryALTinterwordspacing
{Texas Instrument}, ``{CC2420: 2.4 GHz IEEE 802.15.4/Zigbee-ready\ RF
  Transceiver}.'' [Online]. Available: \url{http://www.ti.com/
  lit/ds/symlink/cc2420.pdf}
\BIBentrySTDinterwordspacing

\end{thebibliography}

\end{document}